\documentclass[final]{colt2020}
\usepackage{times}
\usepackage{our_macros}

\begin{document}

\title{From tree matching to sparse graph alignment}

\coltauthor{\Name{Luca Ganassali} \Email{luca.ganassali@inria.fr}\\
\addr INRIA, DI/ENS, PSL Research University, Paris, France.
\AND
\Name{Laurent Massouli\'e} \Email{laurent.massoulie@inria.fr}\\
\addr MSR-Inria Joint Centre, INRIA, DI/ENS, PSL Research University, Paris, France.
}

\maketitle

\begin{abstract}
In this paper we consider alignment of sparse graphs, for which we introduce the \textbf{\textit{Neighborhood Tree Matching Algorithm}} (NTMA). For correlated Erd\H{o}s-R{\'{e}}nyi random graphs, we prove that the algorithm returns -- in polynomial time -- a positive fraction of correctly matched vertices, and a vanishing fraction of mismatches. This result holds with average degree of the graphs in $O(1)$ and correlation parameter $s$ that can be bounded away from 1, conditions under which random graph alignment is particularly challenging. As a byproduct of the analysis we introduce a matching metric between trees and characterize it for several models of correlated random trees. These results may be of independent interest, yielding for instance efficient tests for determining whether two random trees are correlated or independent.
\end{abstract}

\begin{keywords}
graph alignment, tree matching, \ER random graphs
\end{keywords}

\section*{Introduction}

Graph alignment consists in finding an injective mapping (matching) $\cS\subset V(G_1)\times V(G_2)$ between the vertex sets of two graphs $G_1$ and $G_2$ such that, for any two matched pairs $(i,u)$, $(j,v)\in\cS$, then occurrence of edge $\lbrace i,j \rbrace$ in $G_1$ tends to correspond to occurrence of edge $\lbrace u,v \rbrace$ in $G_2$. When this correspondence is exact for all pairs $(i,u)$, $(j,v)$ of matches, then the subgraphs of $G_1$, $G_2$ induced by their nodes appearing in $\cS$ are isomorphic. In this sense, graph alignment is the search for approximate graph isomorphisms.
It has many applications, among which: social network de-anonymization (\cite{narayanan2008robust}, \cite{DBLP:conf/sp/NarayananS09}), analysis of protein-protein interaction graphs (\cite{DBLP:journals/bmcbi/KazemiHGM16}, \cite{feizi19}), natural language processing (\cite{DBLP:journals/tkdd/BayatiGSW13}), medical image processing (\cite{DBLP:conf/ipmi/LombaertSS13}).

A recent thread of research (\cite{DBLP:conf/kdd/PedarsaniG11}, \cite{DBLP:conf/sigmetrics/DaiCKG19} \cite{DBLP:conf/sigmetrics/CullinaK16}, \cite{DBLP:journals/corr/abs-1811-07821}, \cite{DBLP:journals/pomacs/CullinaKMP19}, \cite{DBLP:journals/corr/abs-1907-08880}, \cite{DBLP:journals/corr/abs-1907-08883}) has investigated the fundamental limits to feasibility of graph alignment in the context of a natural generative model of correlated graphs, namely the correlated \ER random graph model $ERC(n,p,s)$. Specifically, it consists of two random graphs $G_1$ and $G_2$ on node set $[n]=\lbrace 1, \ldots,n \rbrace$ obtained as follows. First generate two aligned graphs $G_1$, $G'_2$ with adjacency matrices $A_1$, $A'_2$ such that for each node pair $\lbrace i,j \rbrace$, one has
\begin{flalign*}
\dP\left(A_1(i,j)=A'_2(i,j)=1\right)&=ps,\\
\dP\left(A_1(i,j)=1,A'_2(i,j)=0\right)&=\dP\left(A_1(i,j)=0,A'_2(i,j)=1\right)=p(1-s),\\
\dP\left(A_1(i,j)=A'_2(i,j)=0\right)&=1-p(2-s),
\end{flalign*} and this independently over node pairs $\lbrace i,j \rbrace$. Graph $G_2$'s adjacency matrix is then $A_2=M_{\sigma} A'_2 \left(M_{\sigma} \right)^{\top}$, where $M_{\sigma}$ is the matrix associated to a permutation $\sigma$ chosen uniformly at random from $\cS_n$.

Researchers have strived to determine for which parameter values $(p,s)$, assuming $n\gg1$, one can recover the unknown permutation $\sigma$, and therefore the alignment $G'_2$ of $G_2$ with $G_1$. As in other high-dimensional inference tasks such as community detection, one expects such goal to be either poly-time achievable, achievable though not in poly-time, or impossible to achieve. The corresponding regions of parameter space are usually referred to as the ``easy'',  ``hard'' or ``Information-theoretically (IT) impossible'' phases for the problem considered.

\cite{DBLP:conf/sigmetrics/CullinaK16} have shown that it is possible to recover $\sigma$ if and only if $nps-\log(n) \underset{n \to \infty}{\longrightarrow} +\infty$ \comm{$s$ and not $1-s$}, thereby characterizing the ``IT-impossible'' phase. \cite{DBLP:journals/corr/abs-1811-07821} have proposed a polynomial-time 'degree profile matching' algorithm, and proven it to recover $\sigma$ under the conditions $np \ge \log^\alpha(n)$, $1-s\le \log^{-\beta}(n)$ for suitable constants $\alpha,\beta>0$, thereby identifying a subset of the ``easy'' phase. More recently, \cite{DBLP:journals/corr/abs-1907-08880}, \cite{DBLP:journals/corr/abs-1907-08883}  have proposed a spectral method, and proven it to recover $\sigma$ under the same conditions.\\

The result of \cite{DBLP:conf/sigmetrics/CullinaK16} shows that there is no hope of recovering $\sigma$, or in other words, of perfectly re-aligning $G_1$ and $G_2$, in the case of sparse graphs, that is graphs with average degree $np$ of order 1. 
Nevertheless, their result does not rule out the possibility of partially recovering the unknown permutation $\sigma$. For the applications mentioned earlier, it is at the same time natural to assume that the graphs involved are sparse, and potentially useful to recover only a fraction of the unknown matches $(i,\sigma(i))$. 

\subsection*{Objectives and main result}
This motivates the present work, whose goal is to  show that \textbf{\textit{partial alignment of sparse correlated graphs is feasible}}, and to introduce a polynomial-time algorithm for producing such partial alignments. Our main result is the proposal of the so-called \textbf{\textit{Neighborhood Tree Matching Algorithm}} (NTMA hereafter) together with the following
\begin{theoreme_principal}
\label{theoreme_principal}
Consider the correlated \ER model $ERC(n,p,s)$, where $p=\lambda/n$. For some $\lambda_0>1$, for all $\lambda\in(1,\lambda_0]$, there exists $s^*(\lambda)<1$ such that, provided $s\in(s^*(\lambda),1]$, the NTMA returns a matching $\cS$ verifying the following properties with high probability:
\begin{equation}
|\cS\cap \{(i,\sigma(i)),\;i\in [n]\}|=\Omega(n),\; |\cS\setminus\{(i,\sigma(i)),\;i\in [n]\}|=o(n).
\end{equation}
\end{theoreme_principal}
In words, our algorithm returns a set of node alignments which contains a negligible fraction of mismatches, and $\Omega(n)$ good matches. Our result covers values of $\lambda$ arbitrarily close to 1, and thus applies to very sparse graphs. For $\lambda<1$, \ER graphs in our correlated model have connected components of size at most logarithmic in $n$, so that there is no hope of recovering a positive fraction of correct matches. This result can be interpreted as follows. For partial graph alignment of sparse \ER correlated random graphs, there is an ``easy phase'' that includes the parameter range $\{(\lambda,s):\lambda\in(1,\lambda_0],\; s\in(s^*(\lambda),1]\}$.

\subsection*{Paper organization}
Description of the Neighborhood Tree Matching Algorithm and the proof strategy for establishing Theorem I are given in Section \ref{section_sparse_graph_alignment}. 
Our algorithm relies essentially on a tree matching operation. To pave the way for Section \ref{section_sparse_graph_alignment}, we therefore introduce in Section \ref{section_tree_matching} a notion of matching weight between trees that is key for our algorithm, and can be computed efficiently in a recursive manner. We further obtain probabilistic guarantees on the matching weights between random trees drawn according to some Galton-Watson branching processes. These are instrumental in the proof of Theorem I. However these may be of independent interest. Indeed we introduce in Section \ref{section_tree_matching} a natural hypothesis testing problem on pairs of random trees, for which we obtain a successful test based on computation of tree matching weights. 

\subsection*{Related work}
Most relevant to the present work are the papers \cite{DBLP:conf/kdd/PedarsaniG11}, \cite{DBLP:conf/sigmetrics/DaiCKG19} \cite{DBLP:conf/sigmetrics/CullinaK16}, \cite{DBLP:journals/corr/abs-1811-07821},
\cite{DBLP:journals/corr/abs-1907-08880}, \cite{DBLP:journals/corr/abs-1907-08883}  already mentioned, which also focus on graph alignment in the context of the correlated \ER model. The main differences between the present paper and these is our focus on sparse random graph models, with average degree $\lambda=O(1)$, our treatment of correlation coefficients $s$ bounded away below 1, and our aim of partial rather than full alignment. Article \cite{DBLP:journals/pomacs/CullinaKMP19} addresses a notion of partial alignment stronger than ours, and hence  requires conditions under which graphs are not sparse.
\cite{DBLP:journals/talg/MakarychevMS14} show that graph alignment is NP-hard to solve, even approximately. This justifies the search for custom algorithms in a variety of scenarios. The main methods proposed are:
\textbf{\textit{Percolation methods}} based on some initial {\em seeds}, i.e. matched node pairs provided a priori (\cite{DBLP:journals/bmcbi/KazemiHGM16}). We remark that the matchings returned by our algorithm could be used as seeds, and then processed e.g. using percolation matching to eventually obtain an improved matching.  \textbf{\textit{Spectral methods}} are considered in \cite{feizi19}, \cite{DBLP:journals/corr/abs-1907-08883}; \textbf{\textit{Degree profile matching}}  is introduced in \cite{DBLP:journals/corr/abs-1811-07821}; \textbf{\textit{ Quadratic programming}} \comm{changed semidefinite to quadratic} approaches are proposed in  \cite{DBLP:journals/pami/ZaslavskiyBV09}. \textbf{\textit{Message passing methods}} are introduced in \cite{DBLP:journals/tkdd/BayatiGSW13}. These are structurally similar to our neighborhood tree alignment approach, which is implemented in a recursive manner and can be seen as a message-passing method. Our algorithm is however different, and comes with novel theoretical guarantees.
Graph alignment is a special case of the quadratic assignment problem, reviewed in \cite{conf/dimacs/PardalosRW93}. 
Database alignment is an important variant of graph alignment, studied in   \cite{DBLP:conf/aistats/DaiCK19}.

\subsection*{Notations}
For a graph  $G$, denote by $V(G)$ its set of vertices, $E(G)$ its set non-oriented edges, and $\overrightarrow{E}(G):=\lbrace (i,j), \lbrace i,j \rbrace \in E(G) \rbrace$ its set of oriented edges. We use the notations $u \sim v$ if $\lbrace u,v \rbrace \in E(G)$ and $u \to v$ if $(u,v)\in \overrightarrow{E}(G)$. The usual graph distance in $G$ will be denoted $\delta_G$. For $v \in V(G)$, let $\mathcal{N}_G(v)$ denote the neighborhood of $v$ in $G$, and $\deg_{G}(v)$ its degree. For $d \ge 1$ we also define $\cB_{G}(v,d)$ the set of vertices at (graph) distance at most $d$ from $v$, and $\cS_{G}(v,d):=\cB_{G}(v,d) \setminus \cB_{G}(v,d-1)$ the set of vertices at distance $d$ from $v$. 
For a rooted tree $\cT$, we let $\rho(\cT)$ denote its root node. For any $i\in V(\cT)\setminus \{\rho(\cT)\}$, we let $\pi_{\cT}(i)$ denote the parent of node $i$ in $\cT$. For $d \ge 1$, we note $\cB_d(\cT)=\cB_{\cT}(\rho(\cT),d)$ and $\cL_d(\cT)=\cS_{\cT}(\rho(\cT),d)$.
We omit the dependencies in $G$ or $\cT$ of these notations when there is no ambiguity.

\section{\label{section_tree_matching} Tree matching}
\label{sec:2}
In this section, we introduce the matching weight between rooted trees and the related matching rate. We then establish bounds on the matching rate for specific models of random trees. We also give an application to a hypothesis testing problem on the independence between two trees. 

\subsection{Matching weight of two rooted trees}
\begin{defi}
For any $d\ge 0$, let $\cA_d$ denote the collection of rooted trees whose leaves are all of depth $d$. Given two rooted trees $\cT$ and $\cT'$, let $\cM_d(\cT,\cT')$ denote the collection of trees $t\in\cA_d$ such that there exist injective embeddings \comm{application to embedding, et rajout de injective (un des reviewers le faisait remarquer pour eviter des matching weights infinis)} $f: V(t)\to V(\cT)$, $f': V(t)\to V(\cT')$ that preserve the rooted tree structure, i.e. such that
$$
\begin{array}{lll}
&f(\rho(t))=\rho(\cT),&f'(\rho(t))=\rho(\cT'),\\
\forall i\in V(t)\setminus\{\rho(t)\},&f(\pi_t(i))=\pi_{\cT}(f(i)),& f'(\pi_t(i))=\pi_{\cT'}(f'(i)).
\end{array}
$$
The \textbf{matching weight of $\cT$ and $\cT'$ at distance $d$} is then defined as:
\begin{equation}
\cW_d(\cT,\cT'):=\sup_{t\in \cM_d(\cT,\cT')}\left| \cL_d(t) \right|,
\end{equation}
i.e. the size, measured in number of leaves, of  the largest tree in  $\cM_d(\cT,\cT')$.
\end{defi}
\begin{figure}[H]
\centering
\vspace{-0.5cm}
    \includegraphics[scale=1]{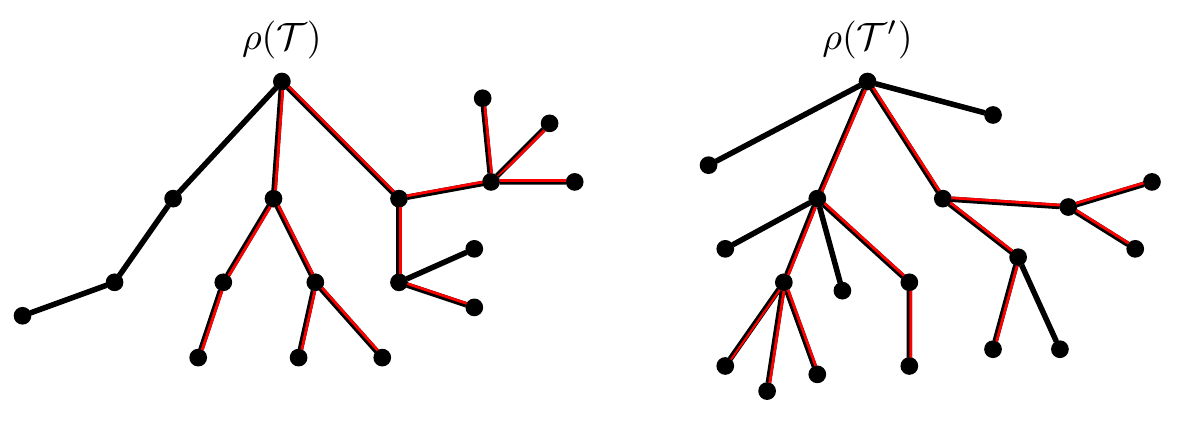}
\caption{\label{image_tree_matching_weight} Example of two trees $\cT$, $\cT'$ with $\cW_3(\cT,\cT')=7$, where an optimal $t \in \cA_3$ is drawn in red.}
\end{figure}
\subsection{Recursive computation of $\cW_d$}
We shall need the following
notations and definitions.
For a tree $\cT$, for $i \in V(\cT)$, and $d \geq 0$, $\cT_d^{(i)}$ is the sub-tree of $\cT$ re-rooted at $i$, containing all vertices at distance less than $d$ of $i$. For $i, j \in V(\cT)$ such that $j \to i$, $\cT_d^{(i \leftarrow j)}$ denotes the sub-tree of $\cT$ re-rooted at $i$, containing all vertices at distance less than $d$ of $i$ but where vertex $j$ has been removed. By definition two vertices  not connected by a path are at distance $\infty$. $\cT_d^{(i \leftarrow j)}$ is thus the tree of depth at most $d$ reached by oriented edge $j \to i$.
\begin{figure}[H]
    \centering
    \includegraphics[scale=1]{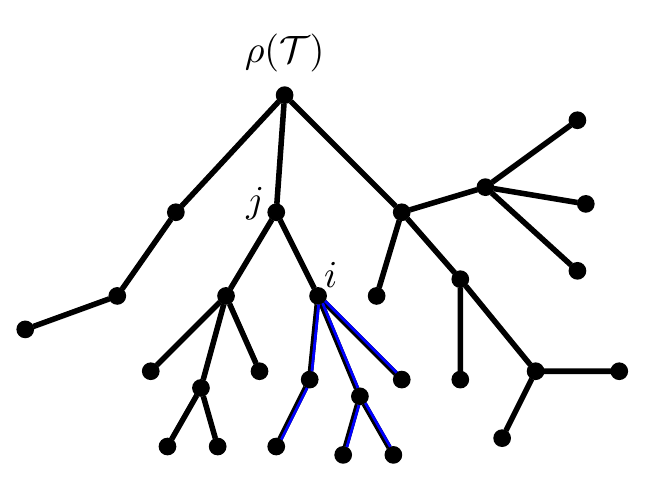}
\caption{\label{image_rerooted_trees} We here show an example of a tree $\cT$ and its corresponding $\cT_2^{(i \leftarrow j)}$ highlighted in blue.}
\end{figure}

\begin{defi}
For a given pair of trees $\cT$ and $\cT'$, for all pair of vertices $(i,u) \in V(\cT) \times V(\cT')$, the \textbf{matching weight of $(i,u)$ at depth $d$} is then defined as:
\begin{equation}
\cW_d(i,u):=\sup_{t\in \cM_d\left(\cT_d^{(i)},\cT_d^{\prime (u)}\right)}\left| \cL_d(t) \right|.
\end{equation}
Moreover, for all pairs of vertices $(i,u),(j,v) \in V(\cT) \times V(\cT')$ such that $j \to i$ and $v \to u$, the \textbf{matching weight of edges $j \to i$ and $v \to u$ at distance $d$} is then defined as:
\begin{equation}\label{eq:weight_cmp}
\cW_d(i \leftarrow j, u \leftarrow v):=\sup_{t\in \cM_d\left(\cT_d^{(i \leftarrow j)},\cT_d^{\prime (u \leftarrow v)}\right)}\left| \cL_d(t) \right|.
\end{equation}
\end{defi}

\begin{remarque}
This definition is compatible with the first one in the context of tree matching: one has $\cW_d(\rho(\cT),\rho(\cT')) = \cW_d(\cT,\cT')$. Note that $\cW_0(i,u)=1$ and $\cW_1(i,u)=\max \left( \deg(i), \deg(u)\right)$. Similarly, $\cW_0(i \leftarrow j, u \leftarrow v)=1$ and $\cW_1(i \leftarrow j, u \leftarrow v)=\max \left( \deg(i),\deg(u)\right)-1$.
\end{remarque}

Now fix $\cT$ and $\cT'$. From these definitions, for all $d \ge 1$, $(i,u),(j,v) \in V(\cT) \times V(\cT')$ such that $j \to i$ and $v \to u$, by doing a first step conditioning, we obtain the following recursion formulae:

\begin{equation}
    \label{rec_formula_W(i,u,j,v)}
    \cW_d(i \leftarrow j, u \leftarrow v) = \sup_{\mathfrak{m} \in \cM\left(\cN_\cT(i) \setminus \lbrace j \rbrace \, , \, \cN_{\cT'}(u) \setminus \lbrace v \rbrace \right)} \sum_{(k,w) \in \mathfrak{m}} \cW_{d-1}(k \leftarrow i, w \leftarrow u),
\end{equation} where $\cM\left(\cE,\cF\right)$ is the set of all maximal injective (or one-to-one) mappings $\mathfrak{m}: \cE_0 \subseteq \cE \to \cF$, where maximal means that they are not restrictions of another one-to-one mapping $\widetilde{\mathfrak{m}}: \cE_1  \to \cF$ with $\cE_1$ such that $\cE \supseteq \cE_1 \supset \cE_0$. In the same way we have 
\begin{equation}
    \label{rec_formula_W(i,u)}
    \cW_d(i,u) = \sup_{\mathfrak{m} \in \cM\left(\cN_\cT(i)\, , \, \cN_{\cT'}(u) \right)} \sum_{(k,w) \in \mathfrak{m}} \cW_{d-1}(k \leftarrow i, w \leftarrow u).
\end{equation}
Thus matching weights at depth $d$ can be obtained by computing weights at depth $d-1$ and solving a linear assignment problem (LAP). 
Recursive formulae (\ref{rec_formula_W(i,u,j,v)}) and (\ref{rec_formula_W(i,u)}) yield simple recursive algorithms (see Algorithms \ref{algo_W_d(i,u,j,v)} and \ref{algo_W_d(i,u)} in \ref{appendix_algos_W}) to  compute all matching weights at depth $d$.

\begin{remarque}\label{rem:complex1}
\comm{a la demande des rewiewers j'ai detaille les explications} The complexity of computing all matching weights at depth $d$ can be obtained as follows. We use dynamic programming and store the $\cW_{k}(i \leftarrow j, u \leftarrow v)$ in a array of size the number of pairs $(e,e’)$ where $e$ and $e’$ are two oriented edges in $\cT$,$\cT’$ (that is, $4 \times |\cT| \times |\cT'|$). Each time we increase $k$, we solve one LAP for each pair $(e,e’)$. The size of the small matrix on which the LAP is done does not exceed $d_\mathrm{max}$, the maximal degree in $\cT$ and $\cT'$. The Hungarian algorithm solves LAP with cubic time complexity. The time complexity is thus $O\left(d \times |\cT| \times |\cT'|\times d_\mathrm{max}^3\right)$. As  $d_\mathrm{max}=O(\log n )$ with high probability, the total complexity is $O(d n^2\log^3 n )$ where $n$ bounds the number of nodes in $\cT$ and $\cT'$. 

However for small values of $n$, the recursive  algorithms \ref{algo_W_d(i,u,j,v)} and \ref{algo_W_d(i,u)} are faster, although of  complexity $O \left( d_\mathrm{max}^{2d}\right)$, which is not polynomial for $d=\Omega(\log n)$.
\end{remarque}
\subsection{Matching rate \comm{change exponent to rate} of random trees}

\begin{defi}
Consider two random trees $\cT, \cT'$. Their \textbf{matching rate} is defined as
\begin{equation}
\label{matching_exponent_eq}
\gamma(\cT,\cT') := \inf \left\lbrace \gamma : \exists m,c,d_0>0, \forall x \geq 0, \forall d \ge d_0, \mathbb{P}\left(\cW_{d}(\mathcal{T},\mathcal{T}') \geq mx \gamma^d\right) \leq e^{-(x-c)_+}\right\rbrace .
\end{equation}
\end{defi}
This quantity captures the geometric rate of growth of matching weights with depth $d$. A simpler alternative definition could be $\widetilde{\gamma}(\cT,\cT') := \inf \left\lbrace \gamma : \mathbb{P}\left(\cW_{d}(\mathcal{T},\mathcal{T}') \geq \gamma^d\right) \underset{d \to \infty}{\longrightarrow} 0\right\rbrace$.  However definition (\ref{matching_exponent_eq}) better suits our purpose.

\begin{remarque}\label{rem:1.5}
By definition,  note that for any $\gamma>\gamma(\cT,\cT')$, $\mathbb{P}\left(\cW_{d}(\mathcal{T},\mathcal{T}') \geq \gamma^d\right)$ converges to $0$ very fast, like $O\left( \exp \left(-c(\gamma)^d\right) \right)$ with $c(\gamma)>1$, so that $\widetilde{\gamma}(\cT,\cT') \leq \gamma(\cT,\cT')$.
\end{remarque}

\subsection{Models of random trees}
We now describe three models of random trees that are relevant to sparse graph alignment. \comm{Un des rewiewers pointait qu'il fallait peut etre definir les modeles d'arbres en precisant les vertices, ou definir des embeddings. Finalement je trouve ça trop long à faire, les dessins donnant des apercus cohérents, et les stats qu'on fait dessus ne faisant jamais intervenir le nom precis des sommets consideres.}

\paragraph{$GW(\lambda)$:} We consider two independent Galton-Watson trees $\cT$ and $\cT'$ with offspring distribution $\Poi( \lambda)$, $\lambda>0$. We denote $(\cT,\cT') \sim GW(\lambda)$.

\paragraph{$GW(\lambda,s,\delta)$:} For $\delta \geq 1$, consider a labeled tree $\cT$ rooted at $\rho$ and a tree $\cT'$ rooted at $\rho'$. $\rho'$ is also a node of $\cT$, at distance $\delta$ from its root $\rho$. The two trees are generated as follows. First, nodes in $\cT$ on the path from $\rho$ to the parent of $\rho'$ in $\cT$ have, besides their child leading to $\rho'$, $\Poi(\lambda)$ children in $\cT$, themselves having offspring in $\cT$ given by independent Galton-Watson trees with offspring $\Poi( \lambda)$. Then, the intersection between $\cT$ and $\cT'$ is a Galton-Watson tree with offspring $\Poi( \lambda s)$, with $\lambda>0$ and $s \in [0,1]$. Then, to each node in the intersection tree, we attach children in $\cT \setminus \cT'$ and children in $\cT' \setminus \cT$, each number being independent $\Poi( \lambda (1-s))$ variables. These children in turn have offspring in the corresponding tree given by independent Galton-Watson trees with offspring $\Poi( \lambda)$. See figure \ref{image_GW} for an illustration. We denote $(\cT,\cT') \sim GW(\lambda,s,\delta)$.

\paragraph{$GW(\lambda,s)$:} It is the previous model with $\delta=0$, so that the two correlated trees $\cT$ and $\cT'$ have same root $\rho$. We denote $(\cT,\cT') \sim GW(\lambda,s)$.

\begin{figure}[H]
\centering
\includegraphics[scale=0.9]{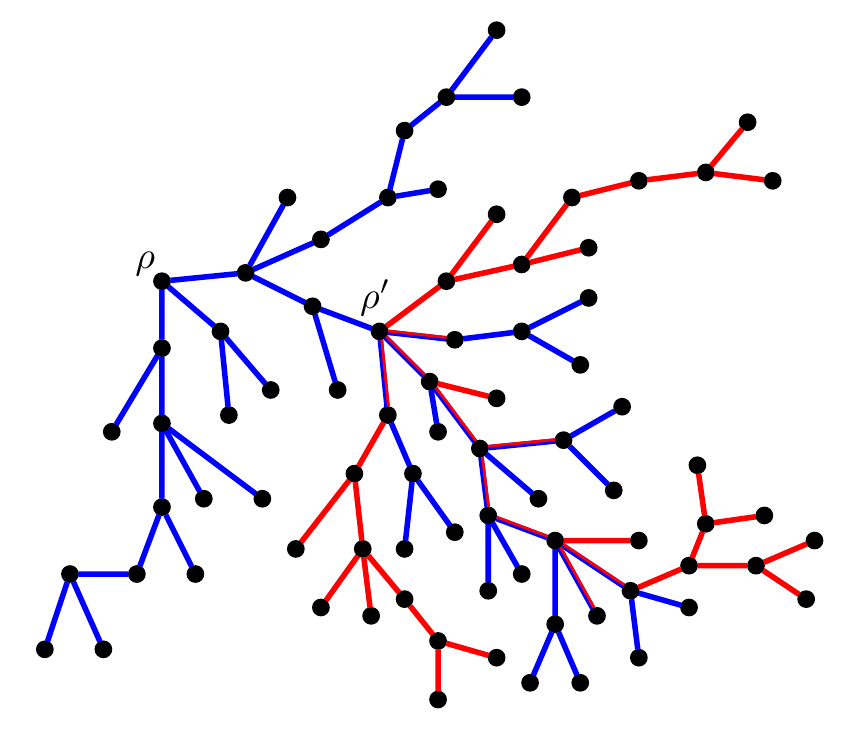}
\caption{\label{image_GW}Random trees $\cT$ (blue) and $\cT'$ (red) from model $GW(\lambda,s,\delta)$ with $\delta=3$.}
\end{figure}
We now turn to the analysis of matching rates for these models.

\begin{prop}
\label{prop_exponent_correlated}
Let $\lambda>1$ and $s \in [0,1]$ such that $\lambda s >1$. For $(\cT,\cT') \sim GW(\lambda,s)$, letting $\gamma(\lambda,s):=\gamma(\cT,\cT')$, we have:
\begin{equation*}
    \gamma(\lambda,s) \geq \lambda s.
\end{equation*}
\end{prop}

\begin{proof}
Let $\cT_{\cap}$ be the intersection tree between $\cT$ and $\cT'$. Branching process theory implies that $(\lambda s)^{-d}\big|\cL_d(\cT_{\cap})\big|$ converges almost surely to a random variable $Z$ as $d\to\infty$, such that $\dP\left(Z>0\right)=1-p$, with $p$ the extinction probability of the branching tree $\cT_{\cap} $. Since $p<1$ when $\lambda s>1$, and for every small enough $\eps>0$, $\lim_{d\to\infty}\dP\left(\cW_d(\cT,\cT')\ge \left(\lambda s(1-\eps)\right)^d\right) \geq 1-p$ \comm{changed the last $=$ into a $\geq$}, the result follows. 
\end{proof}

\subsection{Matching rate of independent Galton-Watson trees}
\begin{theoreme}\label{lambda_close_to_1}
Let $\cT$, $\cT'$ be two independent Galton-Watson trees from the model $GW(\lambda)$. Let $\gamma(\lambda):=\gamma(\cT,\cT')$. There exists $\lambda_0>1$ such that for all $\lambda \in (1,\lambda_0]$, we have
\begin{equation}
\gamma(\lambda)<\lambda.
\end{equation}
\end{theoreme}
Evaluations of $\gamma(\lambda)$ by simulations, confirming the Theorem,  are provided in Appendix \ref{app:tree}.

\paragraph{Proof outline} \comm{Mettre en valeur la notation $\cT_d$ qui nous suivra longtemps.}The full proof of Theorem \ref{lambda_close_to_1} is detailed in the appendix (\ref{appendix_proof_lambda_close_to_1}), but we here give the key steps. 

We introduce some notations. First, for a tree $t$, let $r_d(t)$ denote the tree obtained by suppressing nodes at depths greater than $d$, and then iteratively pruning leaves of depth strictly less than $d$. When computing $\cW_d(t,t')$, the only informative sub-trees are precisely $r_d(t)$ and in $r_d(t')$, one of these being empty if $t$ or $t'$ doesn't survive up to depth $d$. In the rest of the paper, we define $\cT_d$ the random variable $r_d(\cT)$ where $\cT$ is conditioned to survive up to depth $d$.

Consider $(\cT,\cT') \sim GW(\lambda)$. We let $\cE_d$ (respectively, $\cE'_d$) denote the event that tree $\cT$ (respectively, $\cT'$) becomes extinct before $d$ generations, i.e. $\cL_d(\cT)=\varnothing$ (respectively, $\cL_d(\cT')=\varnothing$). We let $p_{d}=\dP(\cE_d) = \dP(\cE'_d)$. It is well known that it satisfies the recursion
$$
p_{0}=0,\; p_{d}=e^{-\lambda(1-p_{d-1})}.
$$  We now establish the following lemma on the structure of $\cT_d$:
\begin{lemme}
\label{lemma_pruning_trees}
For any $\lambda>1$, $\cT_d$ can be constructed by first sampling the number of children $D$ of the root $\rho(\cT)$ according to distribution
\begin{equation*}
    \dP(D=k)=\mathbf{1}_{k>0}\frac{\dP(\Poi(\lambda(1-p_{d-1}))=k)}{\dP(\Poi(\lambda(1-p_{d-1}))>0)}=: q_{d,k},
\end{equation*}
and then attaching $D$ independent copies of $\cT_{d-1}$ to the $D$ children of $\rho(\cT)$.
\end{lemme}
Assume $\eps = \lambda - 1$ to be small enough \comm{added this first sentence}. Fix $r \in (0,1)$, let $\gamma=1+r\eps$. We first show using exponential moments that there exist $m,c>0$ and $d_0 >0$ such that for all $x>0$ $$\dP \left(\cW_{d_0}\left(\cT_{d_0},\cT'_{d_0}\right) \geq m x\right) \leq e^{-x+c}.$$
Then we define the random variables
$$X_d := \gamma^{-(d-d_0)} m^{-1} \cW_d\left(\cT_{d},\cT'_{d}\right).$$
Then, considering the number $D$ of children of the root in $\cT_{d}$ (resp. $D'$ in $\cT'_{d}$), using the previous lemma, one can establish, for all $x>0$, a recursive formula of the following form
\begin{equation*}
    \dP\left( X_d \geq x \right) \leq \sum_{k,\ell \geq 1} q_{d,k} q_{d,\ell} \dP \left(\exists \mathfrak{m} \in \cM\left([k],[\ell]\right), \; \sum_{(i,u) \in \mathfrak{m}} X_{d-1,i,u} \geq \gamma x \right),
\end{equation*} where the $X_{d-1,i,u}$ are i.i.d. copies of $X_{d-1}$. The union bound yields
\begin{equation*}
    \dP\left( X_d \geq x \right) \leq \sum_{k,\ell \geq 1} q_{d,k} q_{d,\ell} \min \left( 1, (k \vee \ell)^{\underline{k \wedge \ell}} \times \dP \left( \sum_{i=1}^{k \wedge \ell} X_{d-1,i,u} \geq \gamma x \right)\right),
\end{equation*} where $m^{\underline{p}}:= m(m-1)\ldots(m-p+1) = \frac{m!}{(m-p)!}$. This inequality enables, with a few more technical steps (see \ref{appendix_proof_lambda_close_to_1}), to propagate recursively the inequality
\begin{equation*}
    \dP\left( X_d \geq x \right) \leq e^{-(x-c)_{+}}.
\end{equation*}

\subsubsection*{Implications for a  hypothesis testing problem}
Let a pair of trees $(\cT,\cT')$ be distributed according to $GW(\lambda)$ under the null hypothesis $\cH_0$, and according to $GW(\lambda,s)$ under the alternative hypothesis $\cH_1$. They are thus independent under $\cH_0$, and correlated under $\cH_1$. Consider the following test:
$$
\hbox{Decide }\cH_0\hbox{ if }\cW_d(\cT,\cT') < \gamma^d,\; \cH_1\hbox{ otherwise.}
$$
Assume that $\gamma(\lambda)<\gamma<\lambda s$. Then in view of Remark~\ref{rem:1.5} and Theorem~\ref{lambda_close_to_1} one has for some $c(\gamma)>1$:
$$
\dP\left(\hbox{decide }\cH_1\big|\cH_0\right)=O\left(e^{-c(\gamma)^d}\right),
$$
thus a super-exponential decay of the probability of false positive (first type error). Conversely, in view of Proposition \ref{prop_exponent_correlated}, noting $\cT_{\cap}$ the intersection tree under $\cH_1$, one has
$$
\dP\left(\hbox{decide }\cH_0\big|\cH_1, \hbox{non-extinction of }\cT_{\cap}\right)=o_d(1).
$$
The false negative probability of this test thus also goes to zero, provided the intersection tree survives. As we will see in the next section, this hypothesis testing problem on a pair of random trees is related to our original graph alignment problem much as the so-called tree reconstruction problem, reviewed in \cite{DBLP:conf/dimacs/Mossel01}, is related to community detection in sparse random graphs (see e.g. \cite{Bordenave15}).

\subsection{Matching rate of intersecting trees}
\begin{theoreme}
\label{lambda_close_to_1_delta}
Let $(\cT,\cT') \sim GW(\lambda,s,\delta)$ with $\delta \geq 1$ and $s \in [0,1]$. Let $\gamma(\lambda,s,\delta):=\gamma(\cT,\cT')$. There exists $\lambda_0>1$ such that for all $\lambda \in (1,\lambda_0]$ we have
\begin{equation}
    \sup_{\delta \geq 1} \gamma(\lambda,s,\delta) < \lambda.
\end{equation}
\end{theoreme} \comm{En effet, ici la preuve ne dépend même pas de $s$ ? Cela peut encore paraître bizarre mais ça roule dans la preuve.}
Evaluations of $\gamma(\lambda,s,\delta)$ by simulations, confirming the Theorem,  are provided in Appendix \ref{app:tree}.
\paragraph{Proof outline} The full proof of Theorem \ref{lambda_close_to_1_delta} is detailed in the appendix (\ref{appendix_proof_lambda_close_to_1_delta}), but we here give the key steps. The proof will again be by induction on $d$, the initial step being established with the same argument as in the proof of Theorem \ref{lambda_close_to_1}. 
$\eps = \lambda -1$ is assumed to be small enough. We fix $r \in (0,1)$, and we let $\gamma=1+r\eps'$. We now work with the random variables
$$X'_d := \gamma^{-(d-d_0)} m^{-1} \cW_d\left(\cT_{d},\cT'_d\right),$$ conditionally on the event that the path from $\rho$ to $\rho'$ survives down to depth $d$ in $\cT$. Then, considering $D$ the number of children of $\rho$ in $\cT_{d}$, $D'$ the number of children of $\rho'$ in $\cT'_{d}$ that are in the intersection tree $\cT_d \cap \cT'_{d}$, and $D''$ the number of children of $\rho'$ in $\cT'_{d} \setminus \cT_d$, we establish for all $x>0$ a recursive formula of the following form
\begin{equation*}
    \dP\left( X'_d \geq x \right) \leq \sum_{k,\ell \geq 1} \dP \left( D'+D''= k, D=\ell\right) \min \left( 1, (k \vee \ell)^{\underline{k \wedge \ell}} \, \dP \left( X'_{d-1} + \sum_{i=1}^{k \wedge \ell-1} X_{d-1,i,u} \geq \gamma x \right)\right),
\end{equation*}  where the $X_{d-1,i,u}$ are i.i.d. copies of $X_{d-1}$ as defined in the proof of Theorem \ref{lambda_close_to_1}. Again, with a few more technical steps (see \ref{appendix_proof_lambda_close_to_1_delta}), we are able to propagate recursively the inequality
\begin{equation*}
    \dP\left( X'_d \geq x \right) \leq e^{-(x-c)_{+}}.
\end{equation*}

\section{\label{section_sparse_graph_alignment} Sparse graph alignment}
We now describe our main algorithm and its theoretical guarantees. For simplicity we assume  that the underlying permutation $\sigma$ is the identity.

\subsection{Neighborhood Tree Matching Algorithm (NTMA), main result}

\comm{avec beacoup de sollicitations des reviewers, il semble important de mettre ici un paragraphe donnant de l'intuition sur notre algo, avant de l'écrire.}
The main intuition for the NTMA algorithm is as follows. In order to distinguish matched pairs of nodes $(i,u)$, we consider their neighborhoods at a certain depth $d$, that are close to Galton-Watson trees. In the case where the two vertices are actual matches, the largest common subtree measured in terms of children at depth (exactly) $d$ is w.h.p. of size $\geq (\lambda s)^d $. However, when the two nodes $i$ and $u$ are sufficiently distant, previous study of matching rates shows that the growth rate of largest common subtree will be $< \lambda s$.  The natural idea is thus to apply the test comparing $\cW_d(i,u)$ to $\gamma^d$ for some well-chosen $\gamma$ to decide whether $i$ is matched to $u$. 

But as the reader may have noticed, testing $\cW_d(i,u) > \gamma^d$ is not enough, because two-hop neighbors would dramatically increase the number of incorrectly matched pairs, making the performance collapse. To fix this, we use the \textbf{dangling trees trick}: instead of just looking at their neighborhoods, we look for the downstream trees from two distinct neighbors $j \neq j'$ of $i$, and $v \neq v'$ of $u$. The trick is now to compare both $\cW_{d-1}(j \leftarrow i,v \leftarrow u)$  and $\cW_{d-1}(j' \leftarrow i,v' \leftarrow u)$ to $\gamma^{d-1}$. This way, even if $i \neq u$ and $i$ and $u$ are close by, the pairs of rooted trees that can be considered will lead to one of the four cases considered and illustrated on Figure \ref{fig_parrallel_construction_bis}, that are settled in the proof of Theorem \ref{no_mismatchs}. \comm{certes cela rajoute du volume mais parait indispensable}

Our algorithm is as follows, where matching tree weights $\cW_{d-1}(j \leftarrow i, v \leftarrow u )$ are defined in \eqref{eq:weight_cmp}:

\begin{algorithm}[H]
\caption{\label{algo_theorique}Neighborhood Tree Matching Algorithm for sparse graph alignment}
\SetAlgoLined

\textbf{Input:} Two graphs $G_1$ and $G_2$ of size $n$, average degree $\lambda$, depth $d$, parameter $\gamma$.

\textbf{Output:} A set of pairs $\cS \subset V(G_1) \times V(G_2)$.

$\mathcal{S} \gets \varnothing$

\For{$(i,u) \in V(G_1) \times V(G_2)$}{
	\If{$\cB_{G_1}(i,d)$ and $\cB_{G_2}(u,d)$ contain no cycle, and $\exists j \neq j' \in \mathcal{N}_{G_1}(i), \exists v \neq v' \in \mathcal{N}_{G_2}(u)$ such that $\cW_{d-1}(j \leftarrow i, v \leftarrow u )> \gamma^{d-1} $ and $\cW_{d-1}(j' \leftarrow i, v'
	\leftarrow u )> \gamma^{d-1} $}{	
	$\cS\gets \cS \cup \left\lbrace (i,u) \right\rbrace $
	}
}
\textbf{return} $\cS$

\end{algorithm}

\begin{remarque}
For $d = \lfloor c \log n \rfloor$, in view of Remark \ref{rem:complex1},  with high probability the complexity of NTMA is $$O \left( \left|V(G_1)\right| \left|V(G_2)\right| (\log n)^2 n^{2c \log \lambda}  d_\mathrm{max}^2  \right) + O \left( \left|E(G_1)\right| \left|E(G_2)\right| (\log n)  d_\mathrm{max}^3  \right),$$ where $d_\mathrm{max}$ is the maximum degree in $G_1$ and $G_2$. In the context of Theorems \ref{lot_of_matchs} and \ref{no_mismatchs} the complexity is then 
$O \left( (\log n)^4 n^{5/2}  \right)$.
\end{remarque}
The two results to follow will readily imply Theorem I.
\begin{theoreme}
\label{lot_of_matchs}
Let $(G_1,G_2) \sim ERC(n,\lambda/n,s)$ be two $s-$correlated \ER graphs such that $\lambda s>1$. Let $d = \lfloor c \log n \rfloor$ with $c \log \left(\lambda\left(2-s\right)\right)<1/2$. Then for $\gamma\in(1,\lambda s)$, with high probability,
\begin{equation}
\label{lot_of_matchs_eq}
\frac{1}{n} \sum_{i=1}^{n} \mathbf{1}_{\lbrace (i,i) \in \cS \rbrace}=\Omega(1).
\end{equation} 
In other words, a non vanishing fraction of nodes is correctly recovered by NTMA (\ref{algo_theorique}).
\end{theoreme}

\begin{theoreme}
\label{no_mismatchs}
Let $(G_1,G_2) \sim ERC(n,\lambda/n,s)$ be two $s-$correlated \ER graphs. Assume that $\gamma_0(\lambda):= \max \left(\gamma(\lambda),\sup_{\delta \geq 1} \gamma(\lambda,s,\delta) \right)<\lambda s$, and that $d = \lfloor c \log n \rfloor$ with $c \log \lambda<1/4$. Then for $\gamma \in(\gamma_0(\lambda),\lambda s)$, with high probability,
\begin{equation}
\label{no_mismatchs_eq}
\mathrm{err}(n):=\frac{1}{n}\sum_{i=1}^{n} \mathbf{1}_{\lbrace \exists u  \neq i, \; (i,u) \in \cS \rbrace}=o(1),
\end{equation} 
i.e. only at most a vanishing fraction of nodes are incorrectly matched by NTMA (\ref{algo_theorique}).
\end{theoreme}

\begin{remarque}
The set $\cS$ returned by the NTMA is not necessarily a matching. Let $\cS'$ be obtained by removing all pairs $(i,u)$ of $\cS$ such that $i$ or $u$ appears at least twice. Theorems \ref{lot_of_matchs} and \ref{no_mismatchs} guarantee that $\cS'$ still contains a non-vanishing number of correct matches and a vanishing number of incorrect matches. Theorem I easily follows. Simulations  of NTMA--2, a  simple variant of  of NTMA, are reported  in Appendix \ref{appendix_simulations_simple_variant}. These confirm our theory, as the algorithm returns many good matches and few mismatches.  
\end{remarque}

\subsection{Proof strategy}
We start by stating Lemmas,  adapted from \cite{Massoulie13} and \cite{Bordenave15} and proven in Appendix \ref{appendix_proof_lemmas_sec_2}, that are instrumental in the proofs of 
 Theorems \ref{lot_of_matchs} and \ref{no_mismatchs}.

\begin{lemme}[Control of the sizes of the neighborhoods]
	\label{control_S}
	Let $G \sim ER(n,\lambda/n)$, $d = \lfloor c \log n \rfloor$ with $c \log \lambda <1$. For all $\gamma>0$, there is a constant $C=C(\gamma)>0$ such that with probability $1-O\left(n^{-\gamma}\right)$, for all $i \in [n]$, $t \in [d]$:
	\begin{equation}
	\label{control_S_eq}
	\left| \cS_{G}(i,t) \right| \leq C (\log n) \lambda^t.
	\end{equation}
\end{lemme}

\begin{lemme}[Cycles in the neighborhoods in an $ER$ graph]
\label{cycles_ER}
Let $G \sim ER(n,\lambda/n)$, $d = \lfloor c \log n \rfloor$ with $c \log \lambda <1/2$. There exists $\eps>0$ such that for any vertex $i \in [n]$, one has
\begin{equation}
\label{cycles_ER_eq}
\mathbb{P}\left(\cB_G(i,d) \mbox{ contains a cycle}\right) = O\left( n^{-\eps}\right).
\end{equation}
\end{lemme}

\begin{lemme}[Two logarithmic neighborhoods are typically size-independent]
\label{indep_neighborhoods}
Let $G \sim ER(n,\lambda/n)$ with $\lambda >1$, $d = \lfloor c \log n \rfloor$ with $c \log \lambda < 1/2 $. Then there exists $\eps>0$ such that for any fixed nodes $i \neq j$, the variation distance between the joint law of the neighborhoods $\mathcal{L} \left(\left(\cS_{G}(i,t),\cS_{G}(j,t)\right)_{t \leq d}\right)$ and the product law $\mathcal{L} \left(\left(\cS_{G}(i,t)\right)_{t \leq d}\right) \otimes \mathcal{L} \left(\left(\cS_{G}(j,t)\right)_{t \leq d}\right)$ tends to $0$ as $O\left(n^{-\eps}\right)$ for some $\eps>0$ when $n \to \infty$.
\end{lemme}

\begin{lemme}[Coupling the $\left|\cS_{G}\left(i,t\right)\right|$ with a Galton-Watson process]
\label{coupling_GW}
Let $G \sim ER(n,\lambda/n)$, $d = \lfloor c \log n \rfloor$ with $c \log \lambda<1/2$. For a fixed $i \in [n]$, the variation distance between the law of $\left( \left|\cS_{G}(i,t)\right|\right)_{t \leq d}$ and the law of $\left( Z_t\right)_{t \leq d}$ where $(Z_t)_t$ is a Galton-Watson process of offspring distribution $\Poi(\lambda)$ tends to 0 as $O\left( n^{-\eps}\right)$ when $n \to \infty$.
\end{lemme}

\subsubsection*{Proof of Theorems \ref{lot_of_matchs} and \ref{no_mismatchs}}
\begin{proof}[Proof of Theorem \ref{lot_of_matchs}]
Define the joint graph $G_{\cup} = G_1 \cup G_2$. For $i \in [n]$, let $M_i$ denote the event that the algorithm matches $i$ in $G_1$ with $i$ in $G_2$, i.e. on which $\mathcal{B}_{G_1}(i,d)$ and $\mathcal{B}_{G_2}(i,d)$ contain no cycle, and $\exists j \neq j' \in \mathcal{N}_{G_1}(i), \exists v \neq v' \in \mathcal{N}_{G_2}(i)$ such that $\mathcal{W}_{d-1}((j,v),(i,u))>\gamma^{d-1}$ and $\mathcal{W}_{d-1}((j',v'),(i,u))>\gamma^{d-1}$. Denote by ${C}_{\cup,i,d}$ the event that there is no cycle in $\mathcal{B}_{G_{\cup}}(i,d)$. 

Arguing as in the proof of Lemma \ref{coupling_GW}, the two neighborhoods  $\mathcal{B}_{G_1}(i,d)$ and $\mathcal{B}_{G_2}(i,d)$ can be coupled with trees distributed as $GW(\lambda,s)$ of Section \ref{section_tree_matching}. However, we will instead consider the intersection graph $G_{\cap} = G_1 \cap G_2$. Obviously,  $G_{\cap} \sim ER(n,\lambda s/n)$. By Lemma \ref{coupling_GW}, the random variables $|\mathcal{S}_{G_{\cap}}(i,t)|$ can be coupled with a Galton-Watson process with offspring distribution $\Poi(\lambda s)$ up to depth $t=d$. Let $P_i$ denote the event that this coupling succeeds. Since $\lambda s>1$, there is a probability $2 \alpha >0$ that this process survives up to depth $d$ and that the first generation has at least two children. Note $S$ this event. On event  $S$, the matching given by the identity on  the intersection tree implies the existence of two neighbors $ j \neq j' \in \mathcal{N}_{G_1}(i)$ and $ v \neq v' \in \mathcal{N}_{G_2}(i)$ such that with high probability $\mathcal{W}_{d-1}(j \leftarrow i, v \leftarrow u)>\gamma^{d-1}$ and $\mathcal{W}_{d-1}(j' \leftarrow i, v' \leftarrow u)>\gamma^{d-1}$, by standard martingale arguments, as in Proposition \ref{prop_exponent_correlated}. This gives the lower bound for $\mathbb{P}(M_i)$:
$$
\mathbb{P}(M_i) \; \geq \mathbb{P}\left({C}_{\cup,i,d} \cap P_i \cap S \right)\; \geq 2\alpha -o(1) > \alpha > 0.
$$
It is easy to see that $G_{\cup} \sim ER(n,\lambda(2-s)/n)$. For $i \neq j \in [n]$, define $I_{i,j}$ the event on which the two neighborhoods of $i$ and $j$ in $G_{\cup}$ coincide with their independent couplings up to depth $d$. By lemma \ref{indep_neighborhoods}, $\mathbb{P}(I_{i,j})=1-o(1)$. Then for $0<\eps<\alpha$ Markov's inequality yields

\begin{flalign}
\mathbb{P}\left(\frac{1}{n} \sum_{i=1}^{n} \mathbf{1}_{\lbrace (i,i) \in \cS \rbrace}<\alpha-\eps\right) & \leq \mathbb{P}\left(\sum_{i=1}^{n} \left(\mathbb{P}(M_i)-\mathbf{1}_{M_i}\right)>\eps n\right)\\
& \leq \frac{1}{n^2 \eps^2} \left(n \mathrm{Var}\left(\mathbf{1}_{M_1}\right)+ n(n-1)\mathrm{Cov}\left(\mathbf{1}_{M_1},\mathbf{1}_{M_2}\right) \right)\\
& \leq \frac{\mathrm{Var}\left(\mathbf{1}_{M_1}\right)}{n \eps^2} + \frac{1-\mathbb{P}\left(I_{1,2}\right)}{ \eps^2} \to 0. 
\end{flalign}

\end{proof}

\paragraph{Proof strategy for Theorem  \ref{no_mismatchs}.} \comm{j'ai donc repoussé la preuve en appendice, et écrit cela à la place}
Consider  two distinct nodes $i$ and $u$. We place ourselves on the event of high probability that $\mathcal{B}_{G_{\cup}}(i,2d)$ has no cycle. On this event, the two neighborhoods $\cB_{G_1}(i,d)$ and $\cB_{G_2}(u,d)$ can be coupled with two trees rooted at $i,u$ respectively.
We then distinguish several cases that are shown on Figure \ref{fig_parrallel_construction_bis}, that require detailed analysis, and which all show that for $i$ fixed, one has 
$$ \mathbb{P}\left(\exists u  \neq i, \; (i,u) \in \cS\right) = o(1).$$
The full proof is deferred to Appendix \ref{app:th2.2}.

\section{Conclusion}
We have introduced NTMA, an algorithm we proved to succeed at partial alignment of sparse correlated random graphs. While our Theorem applies to a limited range of average degrees $\lambda$, we conjecture that rates $\gamma(\lambda)$ and $\gamma(\lambda,s,\delta)$ are strictly less than $\lambda$ for all $\lambda>1$ and $s<1$ and thus NTMA in fact succeeds for a much broader parameter range. This will be the object of future work. 
\newpage

\section*{Acknowledgments}
This paper was partially supported by the Paris Artificial Intelligence Research Institute (PRAIRIE). 

\bibliography{BibCommunityDetection,biblio_sparse_alignment}

\begin{thebibliography}{22}
\providecommand{\natexlab}[1]{#1}
\providecommand{\url}[1]{\texttt{#1}}
\expandafter\ifx\csname urlstyle\endcsname\relax
  \providecommand{\doi}[1]{doi: #1}\else
  \providecommand{\doi}{doi: \begingroup \urlstyle{rm}\Url}\fi

\bibitem[{Barbour} and {Chen}(2005)]{Barbour05}
A.~D. {Barbour} and Louis H.~Y. {Chen}.
\newblock \emph{An Introduction to Stein's Method}.
\newblock co-published with Singapore University, 2005.
\newblock \doi{10.1142/5792}.
\newblock URL \url{https://www.worldscientific.com/doi/abs/10.1142/5792}.

\bibitem[Bayati et~al.(2013)Bayati, Gleich, Saberi, and
  Wang]{DBLP:journals/tkdd/BayatiGSW13}
Mohsen Bayati, David~F. Gleich, Amin Saberi, and Ying Wang.
\newblock Message-passing algorithms for sparse network alignment.
\newblock \emph{{TKDD}}, 7\penalty0 (1):\penalty0 3:1--3:31, 2013.
\newblock \doi{10.1145/2435209.2435212}.
\newblock URL \url{https://doi.org/10.1145/2435209.2435212}.

\bibitem[Bordenave et~al.(2015)Bordenave, Lelarge, and
  Massouli{\'{e}}]{Bordenave15}
Charles Bordenave, Marc Lelarge, and Laurent Massouli{\'{e}}.
\newblock Non-backtracking spectrum of random graphs: Community detection and
  non-regular ramanujan graphs.
\newblock In \emph{{IEEE} 56th Annual Symposium on Foundations of Computer
  Science, {FOCS} 2015, Berkeley, CA, USA, 17-20 October, 2015}, pages
  1347--1357, 2015.
\newblock \doi{10.1109/FOCS.2015.86}.
\newblock URL \url{https://doi.org/10.1109/FOCS.2015.86}.

\bibitem[Cullina and Kiyavash(2016)]{DBLP:conf/sigmetrics/CullinaK16}
Daniel Cullina and Negar Kiyavash.
\newblock Improved achievability and converse bounds for Erd\H{o}s-R\'enyi
  graph matching.
\newblock In \emph{Proceedings of the 2016 {ACM} {SIGMETRICS} International
  Conference on Measurement and Modeling of Computer Science, Antibes
  Juan-Les-Pins, France, June 14-18, 2016}, pages 63--72, 2016.
\newblock \doi{10.1145/2896377.2901460}.
\newblock URL \url{https://doi.org/10.1145/2896377.2901460}.

\bibitem[Cullina et~al.(2019)Cullina, Kiyavash, Mittal, and
  Poor]{DBLP:journals/pomacs/CullinaKMP19}
Daniel Cullina, Negar Kiyavash, Prateek Mittal, and H.~Vincent Poor.
\newblock Partial recovery of Erd\H{o}s-R\'enyi graph alignment via k-core
  alignment.
\newblock \emph{{POMACS}}, 3\penalty0 (3):\penalty0 54:1--54:21, 2019.
\newblock \doi{10.1145/3366702}.
\newblock URL \url{https://doi.org/10.1145/3366702}.

\bibitem[Dai et~al.(2019{\natexlab{a}})Dai, Cullina, and
  Kiyavash]{DBLP:conf/aistats/DaiCK19}
Osman~Emre Dai, Daniel Cullina, and Negar Kiyavash.
\newblock Database alignment with gaussian features.
\newblock In \emph{The 22nd International Conference on Artificial Intelligence
  and Statistics, {AISTATS} 2019, 16-18 April 2019, Naha, Okinawa, Japan},
  pages 3225--3233, 2019{\natexlab{a}}.
\newblock URL \url{http://proceedings.mlr.press/v89/dai19b.html}.

\bibitem[Dai et~al.(2019{\natexlab{b}})Dai, Cullina, Kiyavash, and
  Grossglauser]{DBLP:conf/sigmetrics/DaiCKG19}
Osman~Emre Dai, Daniel Cullina, Negar Kiyavash, and Matthias Grossglauser.
\newblock Analysis of a canonical labeling algorithm for the alignment of
  correlated erd{\H{o}}s-r{\'{e}}nyi graphs.
\newblock In \emph{Abstracts of the 2019 SIGMETRICS/Performance Joint
  International Conference on Measurement and Modeling of Computer Systems,
  Phoenix, AZ, USA, June 24-28, 2019}, pages 97--98, 2019{\natexlab{b}}.
\newblock \doi{10.1145/3309697.3331505}.
\newblock URL \url{https://doi.org/10.1145/3309697.3331505}.

\bibitem[Ding et~al.(2018)Ding, Ma, Wu, and
  Xu]{DBLP:journals/corr/abs-1811-07821}
Jian Ding, Zongming Ma, Yihong Wu, and Jiaming Xu.
\newblock Efficient random graph matching via degree profiles.
\newblock \emph{CoRR}, abs/1811.07821, 2018.
\newblock URL \url{http://arxiv.org/abs/1811.07821}.

\bibitem[Fan et~al.(2019{\natexlab{a}})Fan, Mao, Wu, and
  Xu]{DBLP:journals/corr/abs-1907-08880}
Zhou Fan, Cheng Mao, Yihong Wu, and Jiaming Xu.
\newblock Spectral graph matching and regularized quadratic relaxations {I:}
  the gaussian model.
\newblock \emph{CoRR}, abs/1907.08880, 2019{\natexlab{a}}.
\newblock URL \url{http://arxiv.org/abs/1907.08880}.

\bibitem[Fan et~al.(2019{\natexlab{b}})Fan, Mao, Wu, and
  Xu]{DBLP:journals/corr/abs-1907-08883}
Zhou Fan, Cheng Mao, Yihong Wu, and Jiaming Xu.
\newblock Spectral graph matching and regularized quadratic relaxations {II:}
  erd{\H{o}}s-r{\'{e}}nyi graphs and universality.
\newblock \emph{CoRR}, abs/1907.08883, 2019{\natexlab{b}}.
\newblock URL \url{http://arxiv.org/abs/1907.08883}.

\bibitem[Feizi et~al.(2019)Feizi, Quon, Mendoza, Medard, Kellis, and
  Jadbabaie]{feizi19}
Soheil Feizi, Gerald Quon, Mariana Mendoza, Muriel Medard, Manolis Kellis, and
  Ali Jadbabaie.
\newblock Spectral alignment of graphs.
\newblock \emph{IEEE Transactions on Network Science and Engineering}, 2019.

\bibitem[Kazemi et~al.(2016)Kazemi, Hassani, Grossglauser, and
  Modarres]{DBLP:journals/bmcbi/KazemiHGM16}
Ehsan Kazemi, Seyed~Hamed Hassani, Matthias Grossglauser, and Hassan~Pezeshgi
  Modarres.
\newblock {PROPER:} global protein interaction network alignment through
  percolation matching.
\newblock \emph{{BMC} Bioinformatics}, 17:\penalty0 527:1--527:16, 2016.
\newblock \doi{10.1186/s12859-016-1395-9}.
\newblock URL \url{https://doi.org/10.1186/s12859-016-1395-9}.

\bibitem[Klenke and Mattner(2009)]{Klencke09}
Achim Klenke and Lutz Mattner.
\newblock Stochastic ordering of classical discrete distributions, 2009.

\bibitem[Lombaert et~al.(2013)Lombaert, Sporring, and
  Siddiqi]{DBLP:conf/ipmi/LombaertSS13}
Herve Lombaert, Jon Sporring, and Kaleem Siddiqi.
\newblock Diffeomorphic spectral matching of cortical surfaces.
\newblock In \emph{Information Processing in Medical Imaging - 23rd
  International Conference, {IPMI} 2013, Asilomar, CA, USA, June 28-July 3,
  2013. Proceedings}, pages 376--389, 2013.
\newblock \doi{10.1007/978-3-642-38868-2\_32}.
\newblock URL \url{https://doi.org/10.1007/978-3-642-38868-2\_32}.

\bibitem[Makarychev et~al.(2014)Makarychev, Manokaran, and
  Sviridenko]{DBLP:journals/talg/MakarychevMS14}
Konstantin Makarychev, Rajsekar Manokaran, and Maxim Sviridenko.
\newblock Maximum quadratic assignment problem: Reduction from maximum label
  cover and lp-based approximation algorithm.
\newblock \emph{{ACM} Trans. Algorithms}, 10\penalty0 (4):\penalty0
  18:1--18:18, 2014.
\newblock \doi{10.1145/2629672}.
\newblock URL \url{https://doi.org/10.1145/2629672}.

\bibitem[{Massouli\'e}(2013)]{Massoulie13}
Laurent {Massouli\'e}.
\newblock {Community detection thresholds and the weak Ramanujan property}.
\newblock \emph{arXiv e-prints}, art. arXiv:1311.3085, Nov 2013.

\bibitem[Mossel(2001)]{DBLP:conf/dimacs/Mossel01}
Elchanan Mossel.
\newblock Survey: Information flow on trees.
\newblock In \emph{Graphs, Morphisms and Statistical Physics, Proceedings of a
  {DIMACS} Workshop, New Brunswick, New Jersey, USA, March 19-21, 2001}, pages
  155--170, 2001.
\newblock \doi{10.1090/dimacs/063/12}.
\newblock URL \url{https://doi.org/10.1090/dimacs/063/12}.

\bibitem[Narayanan and Shmatikov(2008)]{narayanan2008robust}
Arvind Narayanan and Vitaly Shmatikov.
\newblock Robust de-anonymization of large sparse datasets.
\newblock In \emph{Proc. of the 29th IEEE Symposium on Security and Privacy},
  pages 111--125. IEEE Computer Society, May 2008.
\newblock \doi{10.1109/SP.2008.33}.
\newblock URL \url{http://www.cs.utexas.edu/~shmat/shmat_oak08netflix.pdf}.

\bibitem[Narayanan and Shmatikov(2009)]{DBLP:conf/sp/NarayananS09}
Arvind Narayanan and Vitaly Shmatikov.
\newblock De-anonymizing social networks.
\newblock In \emph{30th {IEEE} Symposium on Security and Privacy (S{\&}P 2009),
  17-20 May 2009, Oakland, California, {USA}}, pages 173--187, 2009.
\newblock \doi{10.1109/SP.2009.22}.
\newblock URL \url{https://doi.org/10.1109/SP.2009.22}.

\bibitem[Pardalos et~al.(1993)Pardalos, Rendl, and
  Wolkowicz]{conf/dimacs/PardalosRW93}
Panos~M. Pardalos, Franz Rendl, and Henry Wolkowicz.
\newblock The quadratic assignment problem: A survey and recent developments.
\newblock In Panos~M. Pardalos and Henry Wolkowicz, editors, \emph{Quadratic
  Assignment and Related Problems}, volume~16 of \emph{DIMACS Series in
  Discrete Mathematics and Theoretical Computer Science}, pages 1--42.
  DIMACS/AMS, 1993.
\newblock ISBN 978-0-8218-6607-8.
\newblock URL
  \url{http://dblp.uni-trier.de/db/conf/dimacs/dimacs16.html#PardalosRW93}.

\bibitem[Pedarsani and Grossglauser(2011)]{DBLP:conf/kdd/PedarsaniG11}
Pedram Pedarsani and Matthias Grossglauser.
\newblock On the privacy of anonymized networks.
\newblock In \emph{Proceedings of the 17th {ACM} {SIGKDD} International
  Conference on Knowledge Discovery and Data Mining, San Diego, CA, USA, August
  21-24, 2011}, pages 1235--1243, 2011.
\newblock \doi{10.1145/2020408.2020596}.
\newblock URL \url{https://doi.org/10.1145/2020408.2020596}.

\bibitem[Zaslavskiy et~al.(2009)Zaslavskiy, Bach, and
  Vert]{DBLP:journals/pami/ZaslavskiyBV09}
Mikhail Zaslavskiy, Francis~R. Bach, and Jean{-}Philippe Vert.
\newblock A path following algorithm for the graph matching problem.
\newblock \emph{{IEEE} Trans. Pattern Anal. Mach. Intell.}, 31\penalty0
  (12):\penalty0 2227--2242, 2009.
\newblock \doi{10.1109/TPAMI.2008.245}.
\newblock URL \url{https://doi.org/10.1109/TPAMI.2008.245}.

\end{thebibliography}
\appendix
\section{\label{appendix_algos_W}Algorithms for matching weights}
We here describe algorithms to compute recursively matching weights $\cW_d(i \leftarrow j, u \leftarrow v)$ and $\cW_d(i , u)$.
\begin{algorithm}[H]
	\caption{\label{algo_W_d(i,u)}$\cW_d(i,u)$}
	\SetAlgoLined
	\uIf{$d=0$}{
	    \textbf{return} 1\;}
	\Else{$\cE \gets \cN_\cT(i) $ \; 
	
	    $\cF \gets \cN_{\cT'}(i)$ \; 
	    \For{$(k,w) \in \cE \times \cF$}{
	        Compute $\cW_{d-1}(k \leftarrow i, w \leftarrow u)$\;}
	    Solve the LAP problem $w^* := \sup_{\mathfrak{m} \in \cM\left(\cE,\cF\right)} \sum_{(k,w) \in \mathfrak{m}} \cW_{d-1}(k \leftarrow i, w \leftarrow u)$\; 
	    
	    \textbf{return} $w^*$\;}
\end{algorithm}
\begin{algorithm}[H]
	\caption{\label{algo_W_d(i,u,j,v)}$\cW_d(i \leftarrow j, u \leftarrow v )$}
	\SetAlgoLined
	\uIf{$d=0$}{
	    \textbf{return} 1\;}
	\Else{$\cE \gets \cN_\cT(i) \setminus \lbrace j \rbrace$ \; 
	
	    $\cF \gets \cN_{\cT'}(i) \setminus \lbrace v \rbrace$ \; 
	    \For{$(k,w) \in \cE \times \cF$}{
	        Compute $\cW_{d-1}(k \leftarrow i, w \leftarrow u )$\;}
	    Solve the LAP problem $w^* := \sup_{\mathfrak{m} \in \cM\left(\cE,\cF\right)} \sum_{(k,w) \in \mathfrak{m}} \cW_{d-1}(k \leftarrow i, w \leftarrow u)$\; 
	    
	    \textbf{return} $w^*$\;}
\end{algorithm}

\section{Simulations}
\subsection{Simulations for tree matching}\label{app:tree}
We here present some simulations of matching rates $\gamma(\lambda)$ (figure \ref{gamma_lambda}) and $\gamma(\lambda,s,\delta)$ for $s=1$ (figure \ref{gamma_lambda_delta}) in order to illustrate Theorems \ref{lambda_close_to_1} and \ref{lambda_close_to_1_delta} and the final conjecture. For these simulations, error bars correspond to one standard deviation.

\begin{figure}
\floatconts
{nomatter1}
{\caption{\label{gamma_lambda} Comparison of $d \log \lambda$ (blue) and $\log \cW_d(\cT,\cT')$ (red) for $\cW_d(\cT,\cT') \sim GW(\lambda)$ conditioned to survive ($100$ iterations).}}
{%
\subfigure[\footnotesize{$\lambda=1.2$, $\log \lambda \sim 0.18$. Red dashed slope $\sim 0.12$.}]{%
\includegraphics[scale=0.47]{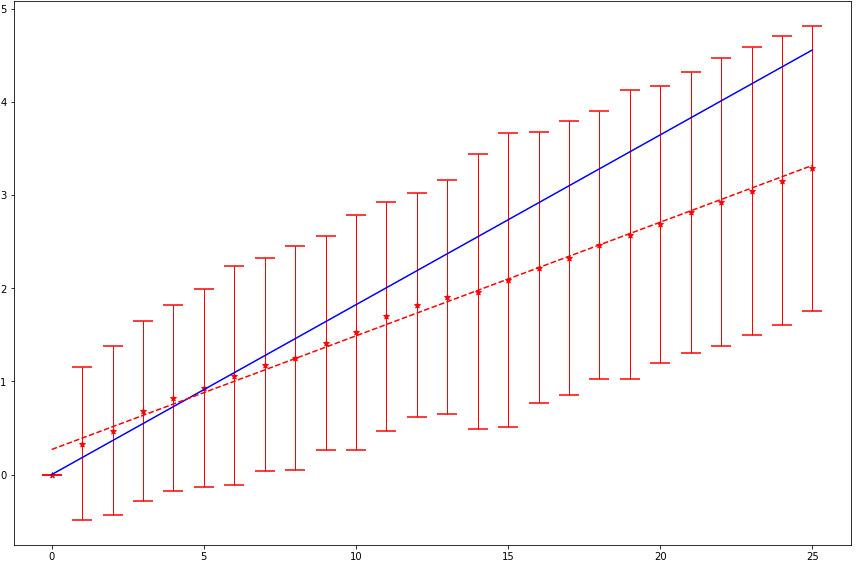}
}

\subfigure[\footnotesize{$\lambda=2.2$, $\log \lambda \sim 0.79$. Red dashed slope $\sim 0.65$.}]{%
\includegraphics[scale=0.47]{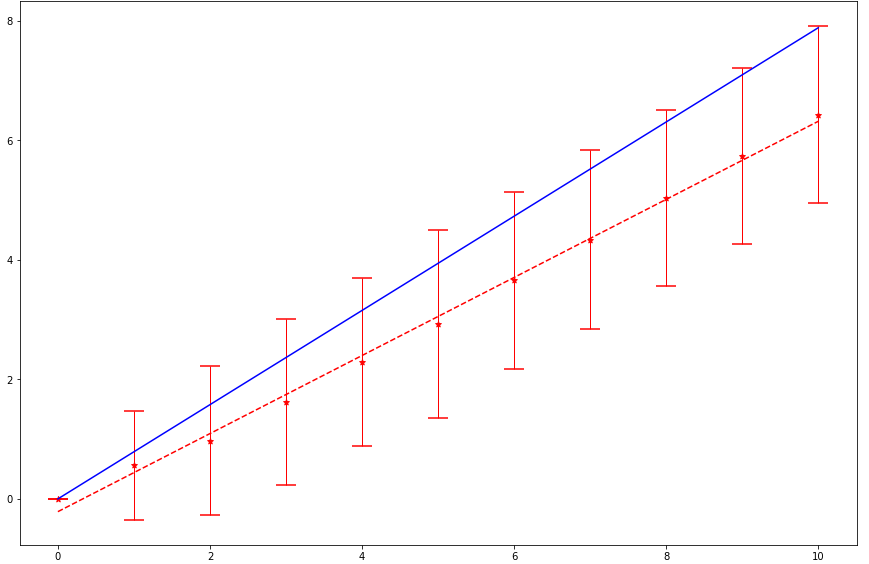}
}

\subfigure[\footnotesize{$\lambda=3.2$, $\log \lambda \sim 1.16$. Red dashed slope $\sim 1.03$.}]{%
\includegraphics[scale=0.47]{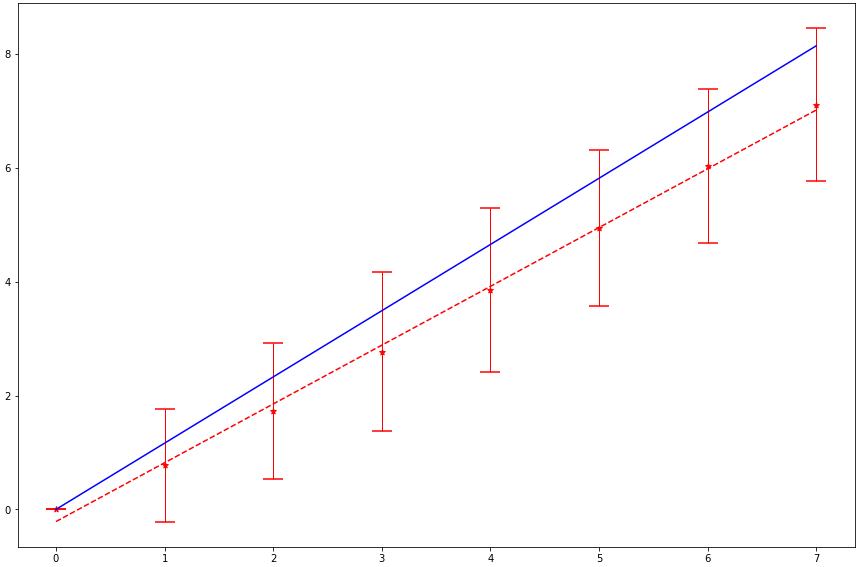}
}
}
\end{figure}

\begin{figure}
\floatconts
{nomatter2}
{\caption{\label{gamma_lambda_delta} Comparison of $d \log \lambda$ (blue) and $\log \cW_d(\cT,\cT')$ (red) for $\cW_d(\cT,\cT') \sim GW(\lambda,s,\delta)$ with $s=1$, conditioned to survive ($50$ iterations).}}
{%
\subfigure[\footnotesize{$\lambda=2.1, \log \lambda \sim 0.74, \delta=1$. Red dashed slope $\sim 0.63$.}]{%
\includegraphics[scale=0.9]{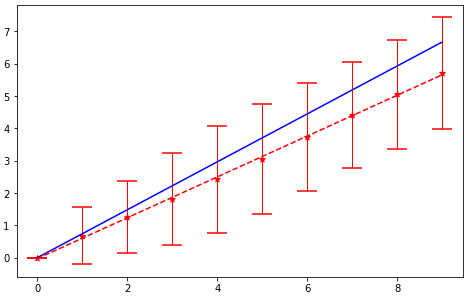}
}

\subfigure[\footnotesize{$\lambda=2.1, \log \lambda \sim 0.74, \delta=2$. Red dashed slope $\sim 0.63$.}]{%
\includegraphics[scale=0.9]{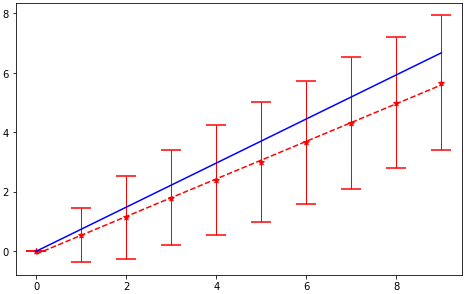}
}

\subfigure[\footnotesize{$\lambda=2.1, \log \lambda \sim 0.74, \delta=5$. Red dashed slope $\sim 0.62$.}]{%
\includegraphics[scale=0.9]{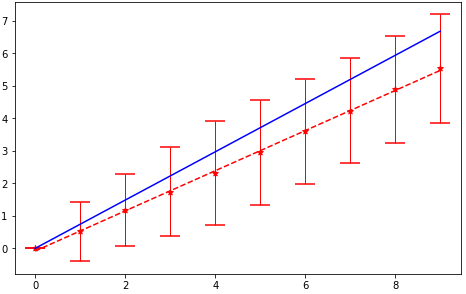}
}

}
\end{figure}

\subsection{\label{appendix_simulations_simple_variant} Simulations for a simple variant algorithm of NTMA}
We here present some simulations of simple variant algorithm of NTMA, NTMA--2, which happens to be more efficient in practice. The Algorithm NTMA--2 is as follows.
\begin{algorithm}[H]
	\caption{NTMA--2}
	\SetAlgoLined
	\textbf{Input:} Two graphs $G_1$ and $G_2$ of size $n$, average degree $\lambda$, depth $d$, parameter $\gamma$.

    \textbf{Output:} A set of pairs $\cS \subset V(G_1) \times V(G_2)$.

	$\cS \gets \varnothing$
		
	\For{$(i,u) \in V(G_1) \times V(G_2)$}{
		\If{$\cW_{d}(i,u)> \gamma^{d} $, $\cW_{d}(i,u)= \max_{j} \cW_{d}(j,u)$ and $\cW_{d}(i,u)= \max_{v} \cW_{d}(i,v)$}{
		$\cS \gets \cS \cup \left\lbrace (i,u) \right\rbrace $}
	}
	\For{$(i,u) \neq (j,v) \in \cS$}{
		\If{$i=j$}{$\cS \gets \cS \setminus \left\lbrace (i,y), y \in V(G_2)\right\rbrace $}
		\If{$u=v$}{$\cS \gets \cS \setminus \left\lbrace (x,u), x \in V(G_1)\right\rbrace $}
	}
	
	\textbf{return} $\cS$
\end{algorithm}
This algorithm only selects rows and columns weight maximums and match the corresponding pairs. The last part ensures that $\cS$ is a matching. For these simulations, error bars correspond to a confidence interval for the mean value of scores. In figures \ref{NTMA2_21} and \ref{NTMA2_31} we compare the scores of NTMA--2 for $s=0.95$ with the isomorphism case $s=1.0$, for different values of $n$. We illustrate the fact that nearly no vertex is mismatched, whereas a non-negligible fraction of nodes is indeed recovered. In figure \ref{sstar}, we compare the scores of NTMA--2 for fixed $n$ but varying $s$, illustrating the existence of a 'critical' parameter $s^*(\lambda)$.

\begin{figure}
\floatconts
{nomatter3}
{\caption{\label{NTMA2_21}Mean score of NTMA--2 for $\lambda=2.1$, $d=5$ (25 iterations per value of $n$).}}
{%
\subfigure[\footnotesize{$s=0.95$.}]{%
\includegraphics[scale=0.55]{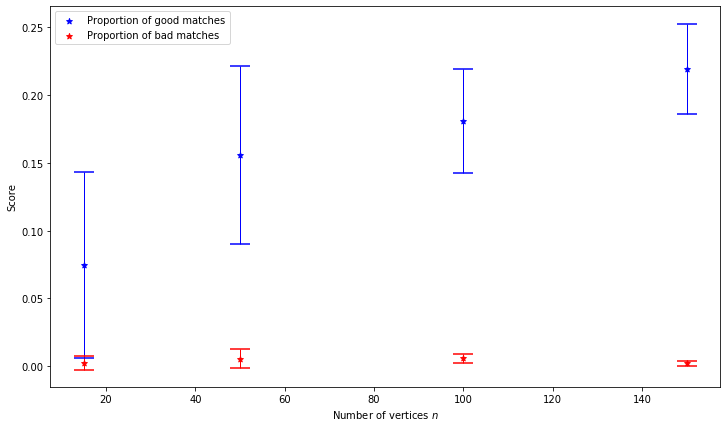}
}

\subfigure[\footnotesize{Isomorphism case , $s=1.0$.}]{%
\includegraphics[scale=0.55]{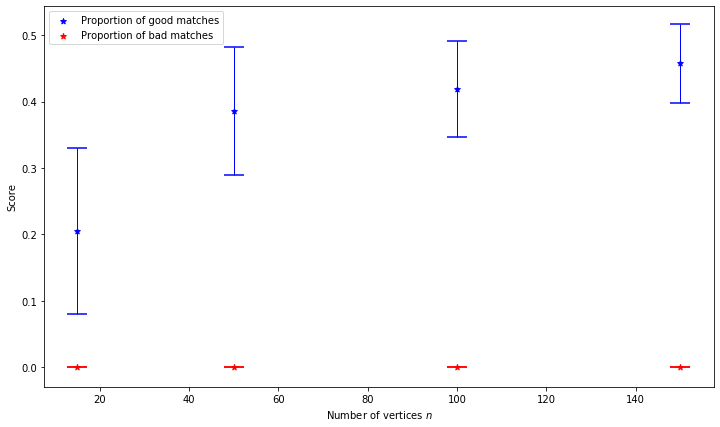} }
}
\end{figure}

\begin{figure}
\floatconts
{nomatter4}
{\caption{\label{NTMA2_31}Mean score of NTMA--2 for $\lambda=3.1$, $d=4$ (25 iterations per value of $n$).}}
{%
\subfigure[\footnotesize{$s=0.95$.}]{%
\includegraphics[scale=0.65]{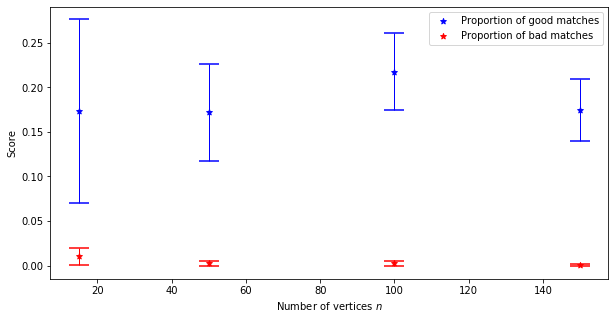}
}

\subfigure[\footnotesize{Isomorphism case , $s=1.0$.}]{%
\includegraphics[scale=0.65]{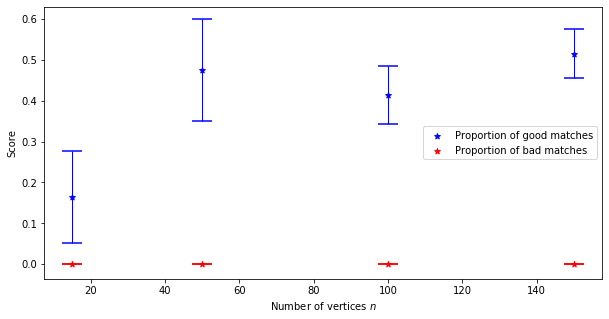}
}
}
\end{figure}

\begin{figure}
\floatconts
{nomatter5}
{\caption{\label{sstar}Mean score of NTMA--2 with different values of $s$ (25 iterations per value of $n$).}}
{%
\subfigure[\footnotesize{$n=150,\lambda=1.4,d=5$.}]{%
\includegraphics[scale=0.55]{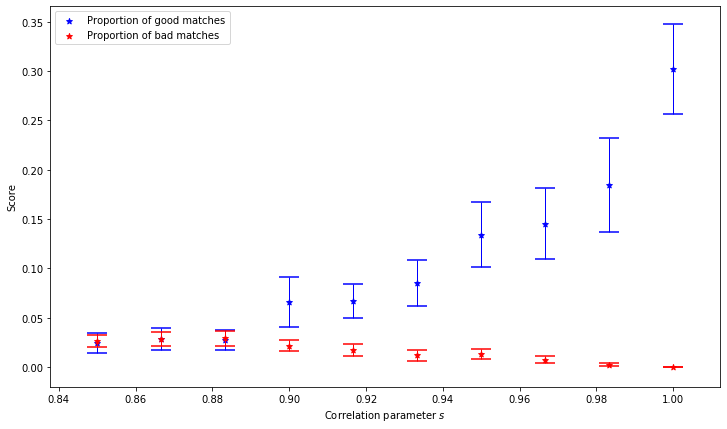} 
}

\subfigure[\footnotesize{$n=50,\lambda=2.2,d=3$.}]{%
\includegraphics[scale=0.55]{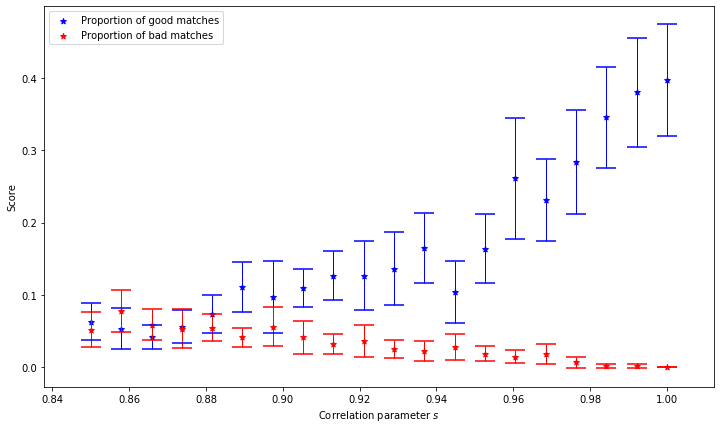}
}
}
\end{figure}

 \section{Detailed proofs for Section \ref{section_tree_matching}}

\subsection{\label{appendix_proof_lambda_close_to_1}Proof of Theorem \ref{lambda_close_to_1}}
\begin{proof}[Proof of Theorem \ref{lambda_close_to_1}]
We first state an easy corollary:
\begin{cor}
\label{cor_exponential_moments}
For any $d \geq 1$, the random variable $X=\left| \cL_d \left( \cT_d \right)\right|$ is such that $\dE\left[e^{\theta X} \right]< \infty$ for all $\theta>0$.
\end{cor}
\begin{proof}
This is easily seen by induction, based on the structure of $\cT_d$ given in Lemma \ref{lemma_pruning_trees}.
\end{proof}
Recall that we let $\cE_d$ (respectively, $\cE'_d$) denote the event that tree $\cT$ (respectively, $\cT'$) becomes extinct before $d$ generations, i.e. $\cL_d(\cT)=\varnothing$ (respectively, $\cL_d(\cT')=\varnothing$). We let $p_{d}=\dP(\cE_d)$. It is well known that it satisfies the recursion
$$p_{0}=0,\; p_{d}=e^{-\lambda(1-p_{d-1})},$$
and converges monotonically to the smallest root in $[0,1]$ of $x=e^{-\lambda(1-x)}$. This root, that we denote $p_e$, is the probability of ultimate extinction. For small enough $\eps=\lambda-1$, it holds that 
$$p_e=1- 2\eps +O(\eps^2),$$ as can be seen by analysis of the fixed point equation satisfied by $p_e$. Let then $d_0$ be such that for all $d\ge d_0$, $p_{d}=1-2\eps+O(\eps^2)$. Clearly, on the event $\cE_d\cup \cE'_d$, the set of matchings $\cM_d(\cT,\cT')$ is empty, so that $\cW_d(\cT,\cT')=0$. Recall that we define $\cT_d$ the random variable $r_d(\cT)$ where $\cT$ is conditioned to survive up to depth $d$.\\
Now fix $r\in(0,1)$. We shall prove that for sufficiently small $\eps>0$, letting $\gamma=1+r \eps$, there exists some constants $c,m,d_0>0$ such that for all $x > 0$, all $d\ge d_0$, one has
\begin{equation} 
\label{exponential_control_eq}
\dP\left(\cW_d(\cT_d,\cT'_d) \ge \gamma^{d-d_0} m x\right)\le e^{-(x-c)_+}.
\end{equation}
We proceed by induction over $d-d_0$. To initialize the induction, notice that one obviously has $\cW_{d_0}(\cT_{d_0},\cT'_{d_0})\le |\cL_{d_0}(\cT_{d_0})|=:X$. By Corollary \ref{cor_exponential_moments}, for all $m,x,\theta>0$, one has:
\begin{flalign*}
\dP\left(  \cW_{d_0}(\cT_{d_0},\cT'_{d_0})>m x \right)\le \dP(X> mx) \le \dE e^{\theta X} e^{-\theta m x}.
\end{flalign*} Let now $\theta=1/m$. By taking $m$ sufficiently large, from dominated convergence we can make $\dE e^{(1/m)X}$ as close to $1$ as we like. Choose for instance $m$ such that $\dE e^{(1/m)X}\le 2$.
Then
\begin{align*}
\dP(\cW_{d_0}(\cT_{d_0},\cT'_{d_0})>m x) \le 2 e^{-x} \leq e^{-x+c}.
\end{align*} for any $c\ge \ln(2)$. Hence, for sufficiently large $m$, we can initialize the induction at $d=d_0$ with any $c \ge \ln(2)$.\\ 

Recall we set $\gamma=1+r\eps$. Define the random variables
$$X_d := \gamma^{-(d-d_0)} m^{-1} \cW_d\left(\cT_{d},\cT'_{d}\right).$$
Let $D$ (resp. $D'$) denote the number of children of the root in $\cT_{d}$ (resp. $D'$ in $\cT'_{d}$). Given $D$ and $D'$, noting $\cT_d=\left(\cT_{d-1,1},\ldots,\cT_{d-1,D}\right)$ and $\cT'_d=\left(\cT'_{d-1,1},\ldots,\cT'_{d-1,D'}\right)$, we have that
\begin{equation*}
\cW_d(\cT_d,\cT'_d)=\sup_{\mathfrak{m} \in \cM([D],[D'])} \sum_{(i,u)\in \mathfrak{m}} \cW_{d-1}(\cT_{d-1,i},\cT'_{d-1,u}),  
\end{equation*} where $\cM([D],[D'])$ denotes the set of all $(D \vee D')^{\underline{D\wedge D'}}$ maximal injective mappings between $\cE_0 \subseteq [D]$ and $[D']$. Let $$X_{d-1,i,u}:=\gamma^{-(d-1-d_0)} m^{-1}\cW_{d-1}(\cT_{d-1,i},\cT'_{d-1,u}).$$ 
Note that conditional on $D$ and $D'$, for each matching $\mathfrak{m} \in \cM([D],[D'])$, the variables $\left( X_{d-1,i,u} \right)_{(i,u)\in \mathfrak{m}}$ are i.i.d. with the same distribution as $X_{d-1}$. The induction hypothesis states that each $X_{d-1,i,u}$ is less, for the strong stochastic ordering of comparison of cumulative distribution functions, than $c$ plus an exponential random variable with parameter 1. With an easy union bound, we can derive the following bounds:

\begin{equation}
\label{eq_key_bound}
\dP\left(X_d>x\right) \leq \sum_{1 \leq k \leq \ell < \infty}\dP\left(D \wedge D'= k, D \vee D'= \ell\right) \min\left( 1, \ell^{\underline{k}} \, \dP \left( \cE_1 + \ldots + \cE_k > \gamma x - kc \right) \right),
\end{equation}
where $\cE_1, \ldots, \cE_k$ are independent exponential random variables of parameter $1$. Lemma \ref{lemma_pruning_trees} states that
\begin{equation*}
\dP(D=k)=e^{-\lambda(1- p_{d-1})}\frac{\lambda^k(1-p_{d-1})^k}{k! \left(1-p_{d}\right)} =:  q_{d,k}.
\end{equation*}We can increase $d_0$ such that for some constant $\kappa>0$, for all $d\ge d_0$:
$$
q_{d,1}\le 1-\eps+\kappa\eps^2, \quad q_{d,k}\le \frac{(3\eps)^{k-1}}{k!},\; k\ge 2.
$$
Note that for $x\le c$, there is nothing to prove in (\ref{exponential_control_eq}), since a probability is always upper-bounded by $1$. We thus only need to consider the case $x>c$. We conclude the proof of this Theorem by appealing to the following

\begin{lemme}
\label{lemma_q_control}
Let $\kappa,C>0$ and $r\in(0,1)$ be given constants. Then there exists $c>0$ large enough and $\eps_0>0$ such that, for all $\eps\in(0,\eps_0)$, letting  $\gamma=1+r\eps$, $q_1=1-\eps+\kappa \eps^2$, $q_k=(C \eps)^{k-1}/k!$ for $k \geq 2$, one has
\begin{equation}
\label{lemma_q_control_eq}
\forall x>c,\quad \sum_{k,\ell \ge 1} q_k q_l \min\left(1, (k \vee \ell)^{\underline{k \wedge \ell}} \, \dP\left(\cE_1 + \ldots + \cE_{k \wedge \ell}>\gamma x-(k \wedge \ell)c\right) \right)\le e^{-(x-c)},
\end{equation}
where the $\cE_i$ are independent exponential random variables of parameter $1$.
\end{lemme}
Its assumptions are indeed verified here with $C=3$, so (\ref{exponential_control_eq}) can be propagated by using this Lemma in (\ref{eq_key_bound}), and the conclusion of Theorem \ref{lambda_close_to_1} follows.

\end{proof}

\subsection{\label{appendix_proof_lambda_close_to_1_delta}Proof of Theorem \ref{lambda_close_to_1_delta}}
\begin{proof}[Proof of Theorem \ref{lambda_close_to_1_delta}]
We assume that $\lambda =1+\eps$. We fix $r \in (0,1)$, and we let $\gamma=1+r\eps$ for some fixed $r\in(0,1)$. \comm{là j'explique un peu le conditionnement}. We work with trees such that $(\cT,\cT') \sim GW(\lambda,s,\delta)$. If we assume that the path from $\rho$ to $\rho'$ does not survive down to depth $d$ in $\cT$, then this path is no more present in $\cT_d$, and the two trees $\cT_d$ and $\cT'_d$ can be coupled with two trees $\widetilde{\cT}_d$ and $\widetilde{\cT}'_d$ where $(\widetilde{\cT},\widetilde{\cT}') \sim GW(\lambda)$, and we are in the case of Theorem \ref{lambda_close_to_1}.

In the following proof, we will thus condition to the event $S_{\rho,d}$ that the path from $\rho$ to $\rho'$ survives down to depth $d$ in $\cT$. Recall that the tree $\cT_d$ (resp. $\cT_d$) is obtained, conditionally on the fact that $\cT$ (resp. in $\cT'$) survives down to depth $d$, by suppressing nodes at depth greater than $d$ in $\cT$ (resp. in $\cT'$), and then pruning alternatively leaves of depth strictly less than $d$. As in the proof of Theorem \ref{lambda_close_to_1}, we shall establish that for sufficiently small $\eps>0$, there exist constants $c,m,d_0>0$ such that for all $x > 0$, all $d\ge d_0$, one has
\begin{equation}
\label{exponential_control_lambda_eq}
\dP\left( \cW_d(\cT_{d},\cT'_d)\ge \gamma^{d-d_0} m x \big| S_{\rho,d} \right) \leq e^{-(x-c)^+}.
\end{equation}
Define the random variables
$$X'_d := \gamma^{-(d-d_0)} m^{-1} \cW_d\left(\cT_{d+\delta},\cT'_d\right),$$ conditional on $S_{\rho,d}$. The proof will again be by induction on $d$, the initial step being established with the same argument as in the proof of Theorem \ref{lambda_close_to_1}. Note that this argument does not depend on $\delta$. 

Denote by $D$ the number of children of $\rho$ in $\cT_{d}$, $D'$ the number of children of $\rho'$ in $\cT'_{d}$ that are in the intersection tree $\cT_d \cap \cT'_{d}$, and $D''$ the number of children of $\rho'$ in $\cT'_{d} \setminus \cT_d$. By branching property, note that these three variables are independent.

Recall that $p_{d}$ denotes the probability that a Galton-Watson tree with offspring $\Poi(\lambda)$ becomes extinct before $d$ generations. Then, conditionally on $S_{\rho,d}$, the random variables $D, D'$ and $D''$ have the following distributions:
\begin{flalign*}
D & \sim 1 + \Poi\left(\lambda\left(1-p_{d-1}\right)\right),\\
D' & \sim \Poi\left(\lambda s \left(1-p_{d-1}\right)\right)\\
D'' & \sim \Poi\left(\lambda(1-s)\left(1-p_{d-1}\right)\right), \mbox{ conditional on } D'+D'' >0.
\end{flalign*} We show an illustration on Figure \ref{image_exemple_D_D'_D''}.
\begin{figure}[H]
\centering
\includegraphics[scale=0.9]{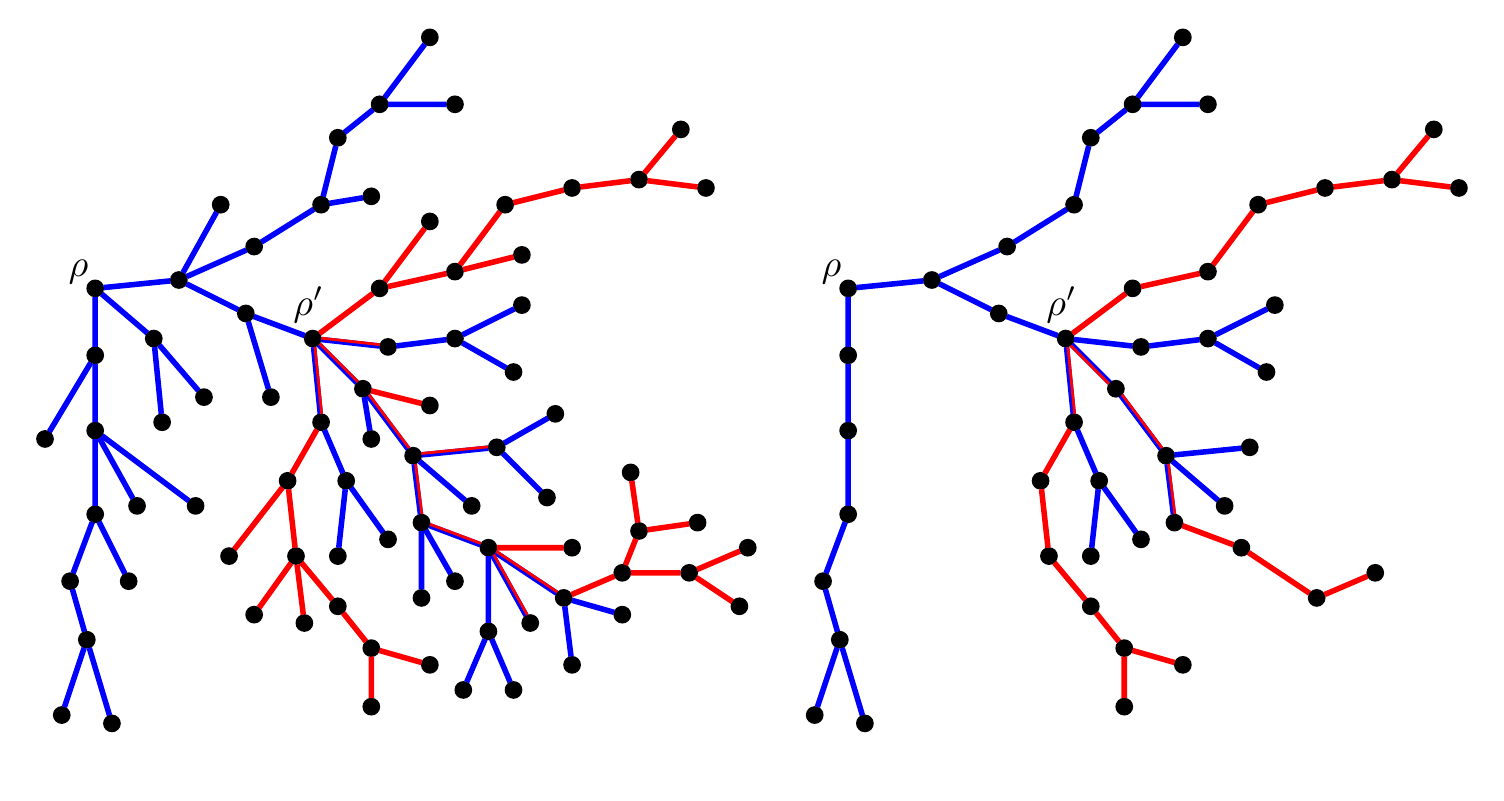}
\caption{\label{image_exemple_D_D'_D''}Random trees $\cT$ (blue) and $\cT'$ (red) from Figure \ref{image_GW} (left), and the results $\cT_d$  and $\cT'_d$ after applying $r_d$ (right). In this example, $\delta=3$ and $d=6$, $D=2$, $D' = 2$ and $D'' = 1$.}
\end{figure}
We condition on the values $\ell,k',k''$ taken by $D,D',D''$. The number of maximal one-to-one mappings between the children of $\rho$ in $\cT_d$ and those of $\rho'$ in $\cT'_d$ is given by $[(k'+k'')\vee \ell]^{\underline{(k'+k'')\wedge \ell}}$, and each of them is of size $\ell\wedge(k'+k'')$. Note here again that for a fixed matching between the children of $\rho$ and $\rho'$, the weights of the matched sub-trees are independent. We distinguish between several cases (to help understand these cases, the reader could keep figure \ref{image_exemple_D_D'_D''} in mind):
\comm{Ou mettre une nouvelle figure, reprenant l'exemple du haut ?}
\begin{itemize}
    \item For a child $i$ of $\rho$ that is not on the path to $\rho'$, the corresponding sub-trees are independent so that the corresponding weight is distributed as $\cW_{d-1}(\widetilde{\cT}_{d-1},\widetilde{\cT}'_{d-1})$ in the independent-tree model $GW(\lambda)$.
    \item If the child of $\rho$ on the path to $\rho'$ is matched with a child of $\rho'$ that is not in the intersection tree, again the corresponding weight is similarly distributed.
    \item Finally, if the child $i$ of $\rho$ leading to $\rho'$ is matched to a child $u$ of $\rho'$ in the intersection tree, setting the new root at $\widetilde{\rho}:=i$ in $\cT$ and at $\widetilde{\rho}':=u$ in $\cT'$, the corresponding weight has the same distribution as $\cW_{d-1}(\cT_{d-1},\cT'_{d-1})$ in the model $GW(\lambda,s,\delta)$, still conditioned to $S_{ \widetilde{\rho},d-1}$. Indeed, there is a path from $\widetilde{\rho}$ to $\widetilde{\rho}'$, and the corresponding Poisson distributions are conserved. 
\end{itemize} The induction hypothesis for case 3, together with Theorem \ref{lambda_close_to_1} for cases 1 and 2, therefore give us:

\begin{equation*}
\dP\left(X'_d\ge x\right) \le \sum_{k,\ell} \dP\left(D'+D''=k,D=\ell\right)\min\left(1, (k \vee \ell)^{\underline{k \wedge \ell}} \, \dP\left(\cE_1 + \ldots + \cE_{k \wedge \ell}>\gamma x-(k \wedge \ell)c\right)\right),
\end{equation*}where the $\cE_i$ are independent exponential random variables of parameter $1$. Assume, as in the proof of Theorem \ref{lambda_close_to_1}, that $d_0$ is chosen such that for all $d\ge d_0$, 
$$
p^{\lambda}_{d}=1-2\eps+O(\eps^2).
$$
With simple computations, we can then ensure that for some $\kappa>0$, noting $q_{d,\cdot}$ the distribution of $D$, one has
$$
q_{d,1} \le 1-\eps+\kappa \eps^2,\quad q_{d,k} \le \frac{(3\eps)^{k-1}}{(k-1)!} \le \frac{(6\eps)^{k-1}}{k!},\; k \geq 2,
$$ where we used $k \leq 2^{k-1}$ in the last step. By independence of $D'$ and $D''$, $D'+D''$ follows a $\Poi(\lambda (1-p_{d-1}))$ distribution, conditional on being positive. Noting $q'_{d,\cdot}$ this distribution, we have, as in the previous proof,
$$
q'_{d,1} \le 1-\eps+\kappa \eps^2,\quad q'_{d,k} \le \frac{(3\eps)^{k-1}}{k!},\; k \geq 2.
$$
We can then invoke Lemma \ref{lemma_q_control} to conclude.
Note that every control in the proof is made uniformly on $\delta \geq 1$ \comm{Et c'est même indépendent de $s$}.
\end{proof}

\subsection{Proof of lemma \ref{lemma_q_control}}
\begin{proof}[Proof of Lemma \ref{lemma_q_control}]. Let
\begin{flalign*}
S_1 & := e^{x-c} q_1^2 e^{-(\gamma x -c)_+} +4 q_1 q_2 e^{-(\gamma x-c)_+},\\
S_2 & := 2 e^{x-c} q_1 \sum_{\ell\ge 3} q_\ell \min\left( 1, \ell e^{-(\gamma x -c)_+}\right),\\
S_3 & := 2 e^{x-c} \sum_{2\le k\le \ell} q_k q_{\ell} \min\left(1,\ell^{\underline{k}} \, \dP\left(\cE_1 + \ldots + \cE_k > \gamma x-kc\right)\right).
\end{flalign*} Our goal is to show that for a suitable choice of $c$, for all $x>c$, $S_1+S_2+S_3\le 1$. One has
\begin{equation}
\label{s1_eq}
S_1 \leq e^{-r \eps x}\left((1-\eps+\kappa \eps^2)^2+2 C \eps \right) \leq e^{-r\eps x}(1+2 C \eps),
\end{equation} and
\begin{equation}
\label{s2_eq}
S_2\le 2 e^{-r\eps  x}(1-\eps +\kappa \eps^2)\sum_{\ell \ge 3}\frac{(C\eps)^{\ell-1}}{(\ell-1)!} \le 2 e^{-r\eps x}\left(e^{C\eps}-1-(C\eps) \right)\le 2 e^{-r\eps x} C^2\eps^2.
\end{equation} We let $k_0$ be such that $\gamma x \in[ k_0c,(k_0+1)c)$. We then upper-bound $S_3$ by $A+B$
where
\begin{flalign}
A & = 2 e^{x-c} \sum_{k=2}^{k_0} q_k \sum_{\ell \ge k} q_{\ell} \frac{\ell!}{(\ell-k)!} \, \dP\left(\cE_1 + \ldots + \cE_k > \gamma x - kc\right), \\
B & = 2 e^{x-c} \sum_{k \ge (k_0+1) \vee 2} \sum_{\ell\ge k} q_k q_{\ell}.
\end{flalign} One readily has
\begin{flalign}
\label{s3_B_eq}
B & \le 2 e^{-r\eps  x} e^{\gamma x-c} \sum_{k\ge (k_0+1)\vee 2} \frac{(C\eps)^{k-1}}{k!} \sum_{\ell\ge k} \frac{(C\eps)^{\ell-1}}{\ell!}\\
&\le 2 e^{-r\eps x} e^{\gamma x -c}\sum_{k\ge (k_0+1)\vee 2}\frac{(C\eps)^{2(k-1)}}{k!}\\
&\le 2 e^{-r\eps x} e^{k_0 c} (C\eps)^{2\left((k_0+1)\vee 2 -1\right)}\\
&\le 2 e^{-r\eps x} \left( C\eps e^{c} \right)^2,
\end{flalign} where in the last steps we assumed that $C\eps e^{c}<1$, so that $$e^{k_0 c} (C\eps)^{2\left((k_0+1)\vee 2 -1\right)} \le (C\eps e^c)^{2\left((k_0+1)\vee 2 -1\right)} \le \left(C\eps e^{c}\right)^2.$$

Note that for $y\ge 0$, $\dP\left(\cE_1 + \ldots + \cE_k > y\right)=\dP(\Poi(y)<k)=e^{-y} \sum_{j=0}^{k-1}y^j/j!$. Write then
\begin{flalign}
\label{s3_A_eq}
\begin{array}{ll}
A & \le 2 e^{x-c} \sum_{k=2}^{k_0} \frac{(C\eps)^{2(k-1)}}{k!} \sum_{j=0}^{k-1} e^{-\gamma x+k c} \frac{(\gamma x)^j}{j!}\\
&\le 2 e^{-r\eps x} \sum_{k=2}^{k_0} \frac{(C^2 e^c)^{k}}{k!} \sum_{j=0}^{k-1}\frac{\left(\gamma x\eps^2\right)^j}{j!}\eps^{2(k-1-j)}\\
&\le 2 e^{-r\eps x}\sum_{k=2}^{k_0}\frac{(C^2 e^c)^{k}}{k!}\left[\eps^{2(k-1)}+e^{\gamma x\eps^2}-1\right]\\
&\le 2 e^{-r\eps x}\left[\eps^{-2}\left(e^{\eps^2 C^2e^c}-1-\eps^2 C^2 e^c \right)  +  e^{C^2 e^c}\left( e^{\gamma x\eps^2}-1  \right)\right]
\end{array}
\end{flalign}
Summing the upper bounds (\ref{s1_eq})-(\ref{s3_A_eq}), the desired property will then hold if for all $x>c$, one has:
\begin{align}
\label{final_eq}
\begin{split}
& e^{-r\eps x}\left[1+2C\eps+ 2C^2\eps^2+2 [C\eps e^{c}]^2+2\eps^{-2}\left(e^{\eps^2 C^2e^c}-1-\eps^2 C^2 e^c \right)\right]\\ & + 2 e^{-r\eps x} e^{C^2 e^c}\left( e^{\gamma x\eps^2}-1  \right)\le 1.
\end{split}
\end{align}

The first term is, for any fixed $c$, and for sufficiently small $\eps$, upper bounded $e^{-r\eps x}\left(1+(2C+1) \eps\right)$.

We now distinguish three cases for $x$.
\begin{itemize}
\item Case 1: $x\in [c,1/\sqrt{\eps}]$. The second term is then $O(\eps\sqrt{\eps})$. Provided $r c>2C+1$, since $e^{-r\eps x}\le e^{-r\eps c}=1-r\eps c+O(\eps^2)$, the left-hand side of (\ref{final_eq}) is then upper-bounded by $1-(rc-2C-1)\eps+O(\eps \sqrt{\eps})$, and is thus less than $1$.
\item Case 2: $x\in [1/\sqrt{\eps},1/\eps]$. Since $e^{-r\eps x}\le e^{-r\sqrt{\eps}}=1-\Omega(\sqrt{\eps})$, and $e^{\gamma x\eps^2}-1\le e^{\gamma \eps}-1=O(\eps)$, the left-hand side of (\ref{final_eq}) is upper-bounded by $1-\Omega(\sqrt{\eps})$ and is thus less than 1.
\item Case 3: $x\ge 1/\eps$. 
The first term is then bounded by $e^{-r}(1+ (2C+1)\eps)$, which is less than $1-\Omega(1)$ for $\eps$ small enough. Letting $y=\eps x$, the second term reads
$$
e^{-ry}[e^{\eps \gamma y}-1](2 e^{C^2 e^c}).
$$
For small $\eps$, this function is maximized for $y=1/r +O(\eps)$, at which point it evaluates to $O(\eps)$. Thus the left-hand side of (\ref{final_eq}) is upper-bounded by $1-\Omega(1)$ in that range.
\end{itemize} 
We have thus shown that for any $r>0$, provided $c> (2C+1)/r$, then for all sufficiently small $\eps$, the desired property holds with $\gamma=1+r\eps$.
\end{proof} 

\subsection{Proof of lemma \ref{lemma_pruning_trees}}
\begin{proof}[Proof of Lemma \ref{lemma_pruning_trees}]
For a tree $t$, we write $t=(t^1,\ldots,t^k)$ to represent the fact that its root has $k$ children, whose offsprings are given by trees $t^1,\ldots,t^k$. Write, noting by $D$ the number of children of $\rho(\cT)$, fixing $k \ge 1$, $t^1,\ldots, t^k\in\cA_{d-1}$, and letting $S=i_1<\cdots<i_k$ run over all $k$ subsets of $[\ell]$:
\begin{flalign*}
\dP(\cT_d=(t^1,\ldots, t^k))&=\sum_{\ell \ge 0}\dP(\cT_d=(t^1,\ldots,t^k),D=\ell)\\
&=\sum_{\ell \ge k}\sum_{S}\dP\left(D=\ell, r_d(\cT^{i_j})=t^j,j\in[k], r_d(\cT^v)=\varnothing, v\notin S\bigg |\overline{\cE}_d\right)\\
&=\frac{1}{1- p_{d}}\sum_{\ell \ge k}\binom{\ell}{k}e^{-\lambda}\frac{\lambda^{\ell}}{\ell!}p_{d-1}^{\ell-k}\prod_{j=1}^k\dP(\cT_{d-1}=t^j) (1-p_{d-1})\\
&=\frac{1}{1- p_{d}}\frac{(\lambda(1-p_{d-1}))^k}{k!}\prod_{j=1}^k\dP(\cT_{d-1}=t^j)\sum_{\ell\ge k}e^{-\lambda}\frac{(\lambda p_{d-1})^{\ell -k}}{(\ell-k)!}\\
&=\frac{e^{-\lambda(1- p_{d-1})}}{1- p_{d}}\frac{(\lambda(1-p_{d-1}))^k}{k!}\prod_{j=1}^k\dP(\cT_{d-1}=t^j)
\end{flalign*}
The conclusion follows by noting that $1-p_{d}=1-e^{-\lambda(1-p_{d-1})}$.
\end{proof}

\section{\label{appendix_proof_lemmas_sec_2}Detailed proofs for Section \ref{section_sparse_graph_alignment}}
The following  proofs are adapted from the previous work of \cite{Massoulie13} and \cite{Bordenave15}.
\subsection{Proof of Lemma \ref{control_S}}
\begin{proof}[Proof of Lemma \ref{control_S}]
	Fix $K>0$ to be specified later and $\gamma>0$. Fix $i \in [n]$, and define 
	\begin{equation*}
	T := \inf \left\lbrace t \leq d, \left|\cS_{G}(i,t)\right| \geq K \log n \right\rbrace.
	\end{equation*} If $T=\infty$, there is nothing to prove. Given $\left| \cS_{G}(i,T-1) \right|$, $$\left| \cS_{G}(i,T) \right| \sim \mathrm{Bin}\left(n-\left| \cS_{G}(i,0) \right|-\ldots-\left| \cS_{G}(i,T-1) \right|,1-\left(1-\frac{\lambda}{n}\right)^{\left| \cS_{G}(i,T-1) \right|}\right).$$ Thus
	$$ \left| \cS_{G}(i,T) \right|  \overset{\mathrm{sto.}}{\leq} \mathrm{Bin}\left(n, \lambda K \frac{\log n}{n}\right).$$
	Using Bennett's inequality, for $K'>\lambda K$:
	\begin{equation*}
	\mathbb{P} \left(\left| \cS_{G}(i,T) \right|  \geq K' \log n \right) \leq e^{- \lambda K h\left(\frac{K'-\lambda K}{\lambda K}\right)\log n },
	\end{equation*} with $h(u)= (1+u) \log(1+u) -u$. This probability is $\leq n^{-2-\gamma}$ if $K'$ is large enough to verify $\lambda K h\left(\frac{K'-\lambda K}{\lambda K}\right) > \gamma +2$. With a simple use of the union bound, one gets that $\left|\cS_{G}(i,T) \right|  \in \left[K \log n, K' \log n\right]$ for all $i \in [n]$ with probability $1-O(n^{-1-\gamma})$.\\
	
	Take $\eps > 0$ to be specified later. We then check by induction that with high probability, for all $T \leq t \leq d$,
	\begin{equation}
	\label{controle_produit_S}
	\left| \cS_{G}(i,t) \right|  \in \left[ K\left(\frac{\lambda}{2}\right)^{t-T}  \left(\log n\right) \prod_{s=T}^{t}\left(1-\eps \left(\frac{\lambda}{2}\right)^{-\frac{s-T}{2}}\right) , K' \lambda^{t-T}  \left(\log n\right) \prod_{s=T}^{t}\left(1+\eps \lambda^{-\frac{s-T}{2}}\right)\right].
	\end{equation}The case $t=T$ is proved here above. We will next use the inequality
	\begin{equation}
	\label{ineq_lemma}
	\lambda u /(2n) \leq \lambda u /n - \lambda^2 u^2 /(2n^2) \leq 1-\left(1-\lambda/n\right)^u \leq \lambda u /n.
	\end{equation} that holds as soon as $\lambda u /n<1$.\\
	
	Assuming (\ref{controle_produit_S}) holds up to $t$, inequality (\ref{ineq_lemma}) holds for $u = \left| \cS_{G}(i,t) \right| $ for $n$ large enough, since $\left| \cS_{G}(i,t) \right|  < n/\lambda$ for $c \log \lambda <1$. Thus for $n$ large enough $\mathbb{E}\left| \cS_{G}(i,t+1) \right| $ lies in the interval 
	\begin{equation*}
	\left[\frac{K}{2} \lambda  \left(\frac{\lambda}{2}\right)^{t-T}  \left(\log n\right) \underbrace{\prod_{s=T}^{t}\left(1-\eps \left(\frac{\lambda}{2}\right)^{-\frac{s-T}{2}}\right)}_{=1-O(\eps)} , \lambda K' \lambda^{t-T}  \left(\log n\right) \prod_{s=T}^{t}\left(1+\eps \lambda^{-\frac{s-T}{2}}\right) \right]
	\end{equation*}
	
	With $\hat{\eps} > 0$ to be specified later, Bennett's inequality writes 
	\begin{equation*}
	\mathbb{P} \left(\bigg|\left| \cS_{G}(i,t+1) \right| - \mathbb{E}\left| \cS_{G}(i,t+1) \right|\bigg|\geq \hat{\eps} \mathbb{E}\left| \cS_{G}(i,t+1) \right|\right) \leq 2 e^{- \frac{1}{2} \lambda  \left(\frac{\lambda}{2}\right)^{t-T}  \log n \left(1-O\left(\eps\right) \right) h\left(\hat{\eps}\right) },
	\end{equation*}
	which is $\leq n^{-2-\gamma}$ if $K \left(\frac{\lambda}{2}\right)^{t+1-T} h(\hat{\eps})>2+\gamma$. Since for $u \to 0$, $h(u)=u^2 / 2 + o(u^2)$, it suffices to take $\hat{\eps}=\eps \left(\frac{\lambda}{2}\right)^{-\frac{t+1-T}{2}}$ with $\eps$ small enough and $K$ large enough such that $K \eps >2+\gamma$. Thus (\ref{controle_produit_S}) holds for $t+1$ with probability $1-O(n^{-2-\gamma})$. \\
	
	All this ensures that the desired inequality (\ref{control_S_eq}) holds for all $i \in [n]$, $t \in [d]$ with probability $1-O(n^{-\gamma})$.
\end{proof}

\subsection{Proof of Lemma \ref{cycles_ER}}
\begin{proof}[Proof of Lemma \ref{cycles_ER}]
Fix $i \in [n]$. Define
$$k^* := \inf \lbrace t \leq d, \; \cB_G(i,t) \mbox{ contains a cycle}\rbrace .$$
Note that $k^* \geq 2$, and that if $k^*=\infty$ then $\cB_G(i,d)$ does not contain any cycle. Now assume that $k^*<\infty$. For any $k \geq 2$, $k^*=k$ if and only if there are two vertices of $\cS_G(i,k-1)$ that are connected, or if there is a vertex of $\cS_G(i,k)$ connected to two vertices of $\cS_G(i,k-1)$. On the event 
$$\mathcal{A}:=\underset{t \leq d}{\bigcap} \left\lbrace \left|\cS_G(i,t)\right|< C (\log n) \lambda^t \right\rbrace, $$
this happens with probability at most $$ \left|\cS_G(i,k-1) \right|^2 \times \frac{\lambda}{n}  + \left|\cS_G(i,k) \right| \times \left|\cS_G(i,k-1) \right|^2 \times \frac{\lambda^2}{n^2} \leq C^2 \frac{(\log n)^2 \lambda^{2k}}{n} + C^3 \frac{(\log n)^3 \lambda^{3k}}{n^2}.$$

Taking $\eps>0$ such that $c \log \lambda \leq 1/2-\eps$, choosing $C$ such that $\mathbb{P}\left(\cA \right)=1-O\left( n^{-2 \eps} \right)$ with Lemma \ref{control_S}, the probability that $\cB_G(i,d)$ contains a cycle is less than
\begin{flalign*}
\mathbb{P}\left( k^* < \infty \right) & \leq  \mathbb{P}\left( \bar{\cA} \right) +  \sum_{k = 2}^{d} \mathbb{P}\left( k^* =k \, | \,  \cA \right)\\
& \leq  O\left( n^{-2 \eps} \right) + O\left( \frac{(\log n)^2 \lambda^{2d}}{n} \right) + O\left( \frac{(\log n)^3 \lambda^{3d}}{n^2} \right) \\
& \leq  O\left( n^{-2 \eps} \right)+O\left( (\log n)^2 n^{-2 \eps} \right)+O\left( (\log n)^3 n^{-3 \eps} \right) \leq O(n^{-\eps}).
\end{flalign*}
\end{proof}

\subsection{Proof of Lemma \ref{indep_neighborhoods}}
\begin{proof}[Proof of lemma \ref{indep_neighborhoods}]
For fixed $i \neq j \in [n]$, let $\left(\tilde{\cS}(i,t)\right)_{t \leq d}$ and $\left(\tilde{\cS}(j,t)\right)_{t \leq d}$ denote two independent realizations of the neighborhoods (i.e. with independent underlying Bernoulli variables). We then construct recursively a coupling $\left(\cS(i,t),\cS(j,t)\right)_{t \leq k} $:
\begin{itemize}
    \item For $k=1$, take $\cS(i,t)$ to be a set of vertices uniformly chosen among sets of $[n]$ of size $\left| \tilde{\cS}(i,0)\right|$. Independently, take $\cS(j,t)$ to be a set of vertices uniformly chosen among sets of $[n]$ of size $\left| \tilde{\cS}(j,0)\right|$.
    
    \item Now if $k>1$, construct $\cS(i,k)$ as follows: select a subset of $[n] \setminus \left(\underset{s \leq k-1}{\bigcup} \cS(i,s)\right)$ of size $ \left|\tilde{\cS}(i,k)\right| $ uniformly at random. Then we construct independently $\cS(j,k)$ taking a uniform subset of $[n] \setminus \left(\underset{s \leq k-1}{\bigcup} \cS(j,s)\right)$ of size $ \left|\tilde{\cS}(j,k)\right| $.
\end{itemize}

This coupling is well defined, and coincides with the independent setting up to step $k$ as long as the sets $\underset{s \leq k}{\bigcup} \cS(i,s)$ and $\underset{s \leq k}{\bigcup} \cS(j,s)$ do not intersect. On the event $$\mathcal{A}:=\underset{t \leq d}{\bigcap} \left\lbrace \left|\cS(i,t)\right|,\left|\cS(j,t)\right| < C (\log n) \lambda^t \right\rbrace, $$ one has
\begin{flalign*}
\mathbb{E}\left[\left|\underset{k \leq d}{\bigcup} \cS(i,s) \cap \underset{k \leq d}{\bigcup} \cS(j,s)\right|\right] & \leq 
\mathbb{E}\left[\sum_{k=1}^{d} \mathrm{Bin}\left(C (\log n) \lambda^k, \frac{\sum_{t=1}^{k} C (\log n) \lambda^t}{n-\sum_{t=1}^{k} C (\log n) \lambda^t} \right) \right] \\
& \leq C^2 (\log n)^2 \left(\frac{\lambda}{\lambda-1}\right) \sum_{k=1}^{d} \frac{\lambda^{2k} }{n-\frac{\lambda}{\lambda-1} C (\log n)\lambda^k}\\
& \leq O \left((\log n)^2 \lambda^{2d}/n\right)
\end{flalign*}if $(\log n) \lambda^d = o(n)$, which is the case if $c \log \lambda <1$. The expectation is upper-bounded by $O \left((\log n)^2 \lambda^{2d}/n\right) = O\left((\log n)^2 n^{-2\eps}\right)$ if $c \log \lambda \leq 1/2 -  \eps$.\\

With Lemma \ref{control_S}, choosing $C$ such that $\mathbb{P}\left(\cA \right)=1-O\left( n^{-2 \eps} \right)$, we get
\begin{flalign*}
\hspace{-0.5cm}
\DTV\left(\mathcal{L} \left(\left(\cS_{G}(i,t),\cS_{G}(j,t)\right)_{t \leq d}\right),\mathcal{L} \left(\left(\cS_{G}(i,t)\right)_{t \leq d}\right) \otimes \mathcal{L} \left(\left(\cS_{G}(j,t)\right)_{t \leq d}\right)\right) & \leq O((\log n)^2 n^{-2 \eps}) + \mathbb{P}\left(\bar{\mathcal{A}}\right)\\
& \leq O(n^{- \eps}).
\end{flalign*}
\end{proof}

\subsection{Proof of Lemma \ref{coupling_GW}}
We work here conditionally on
$$\mathcal{A}:=\underset{t \leq d}{\bigcap} \left\lbrace \left|\cS_G(i,t)\right|< C (\log n) \lambda^t \right\rbrace.$$
Let's define a Galton-Watson process as follows: set $Z_0=1$, and for $t > 0$, $\cL\left( Z_{t} | \cG_{t-1} \right)=\Poi \left( \lambda Z_{t-1}\right)$, where $\cG_t=\sigma\left(Z_s,s \leq t \right)$. 
Fix $t>0$. Conditionally on $\cF_{t-1}:=\sigma\left(\left|\cS_G(i,s)\right|, s \leq t-1 \right)$, define a random variable $W_t$ with distribution $\Poi \left( \lambda \left|\cS_G(i,t-1)\right|\right)$. Note that
$$\cL\left( \left|\cS_G(i,t)\right| \bigg| \cF_{t-1} \right) = \Bin \left( n-\left|\cS_G(i,0)\right|-\ldots-\left|\cS_G(i,t-1)\right|, \; 1-\left(1-\frac{\lambda}{n} \right)^{\left|\cS_G(i,t-1)\right|} \right).$$
The Stein-Chen method (see e.g. \cite{Barbour05}) enables to bound  $\DTV \left(\Bin(n,\lambda/n),\Poi(\lambda)\right)$ by $\min(1,\lambda^{-1})\lambda^2/n \leq \lambda/n$. We also use the classical bound $\DTV \left(\Poi(\lambda),\Poi(\lambda')\right)\leq \left| \lambda-\lambda' \right|$ together with inequality (\ref{ineq_lemma}) (which holds for $n$ large enough since $c \log \lambda <1)$ to obtain that conditionally on $\cF_{t-1}$:
\begin{flalign*}
\DTV\left(\left|\cS_G(i,t)\right|,W_t\right) & \leq n^{-1}  \left(n-\left|\cS_G(i,0)\right|-\ldots-\left|\cS_G(i,t-1)\right|\right) \frac{\lambda \left|\cS_G(i,t-1) \right|}{n} \\
&+ \left| \left(n-\left|\cS_G(i,0)\right|-\ldots-\left|\cS_G(i,t-1)\right|\right) \left( 1-\left(1-\frac{\lambda}{n} \right)^{\left|\cS_G(i,t-1)\right|} \right) - \lambda \left|\cS_G(i,t-1) \right| \right| \\
& \leq \frac{\lambda \left|\cS_G(i,t-1) \right|}{n} + \lambda \left|\cS_G(i,t-1) \right|\\&- \left(n-\left|\cS_G(i,0)\right|-\ldots-\left|\cS_G(i,t-1)\right|\right) \frac{\lambda \left|\cS_G(i,t-1) \right|}{n}  + \frac{\lambda^2 \left|\cS_G(i,t-1) \right|^2}{2n}.
\end{flalign*}
Now, for $\eps>0$ such that $c \log \lambda \leq 1/2 - \eps$, on the event $\cA$, all variables $\left|\cS_G(i,s) \right|$ are bounded by $C (\log n) n^{1/2 - \eps}$. This leads to
\begin{flalign*}
\DTV\left(\left|\cS_G(i,t)\right|,W_t\right) & \leq O\left( (\log n)n^{-1/2-\eps} \right) + O\left( (\log n)^3 n^{-2\eps} \right) + O\left( (\log n)^2 n^{-2\eps} \right) \\& = O\left( (\log n)^3 n^{-2\eps} \right).
\end{flalign*}
This proves by induction that the total variation distance between $\left( \left|\cS_G(i,t)\right| \right)_{t \leq d}$ and $\left( Z_t \right)_{t \leq d}$ is bounded by $O\left( (\log n)^4 n^{-2\eps} \right) = O\left(n^{-\eps}\right)$, taking $C$ large enough in Lemma \ref{control_S} so that $\mathbb{P}\left(\cA \right) \geq 1 - O\left(n^{-2\eps}\right)$.

\subsection{Proof of Theorem \ref{no_mismatchs}}

\begin{proof}[Proof of Theorem \ref{no_mismatchs}]\label{app:th2.2}
Define
\begin{equation*}
d_{\mathrm{max}}:=\max\left(\max_i \mathrm{deg}_{G_1}(i), \max_u \mathrm{deg}_{G_2}(u)\right).
\end{equation*} 
We use the same notations as in the former proof: $G_{\cup} = G_1 \cup G_2$ and $G_{\cap} = G_1 \cap G_2$. Fix $i \in [n]$. In the rest of the proof we work conditionally to the event $C_{\cup,i,2d}$ that $\mathcal{B}_{G_{\cup}}(i,2d)$ has no cycle. Since $c \log \lambda<1/4$, $\dP\left( C_{\cup,i,2d} \right)=1-o(1)$ by Lemma \ref{cycles_ER}. \\
Fix another vertex $u \neq i$. The $d-$neighborhoods $\cB_{G_1}(i,d)$ and $\cB_{G_2}(u,d)$ have offspring distribution stochastically dominated by $\mathrm{Bin}(n,\lambda/n)$, which is also dominated by $\mathrm{Poi}(\lambda')$ as soon as $\lambda'=\lambda+O(1/n)$ (see e.g. \cite{Klencke09}). We can choose $\lambda'$ such that $\gamma > \gamma(\lambda',0)$ still holds: indeed, by a standard coupling argument, one can see that $\gamma : \lambda \mapsto \gamma(\lambda)$ is increasing. We now build two dominating (in the usual edge presence sense) tree-like $d-$neighborhoods of $i$ and $u$ with the following construction. 
\begin{itemize}
    \item First, if the two neighborhoods don't intersect, just sample two independent trees from model $GW(\lambda')$ rooted in $i$ and in $u$. 
    \item If the two neighborhoods intersect, condition to the event that $\alpha$ is the contact point in the path $\mathfrak{p}_{\cup}$ (unique by conditioning on $C_{\cup,i,2d}$) from $i$ to $u$ in the joint graph. Then there is a path of edges of $G_1$ (say, blue) from $i$ to $\alpha$, then a path of edges of $G_2$ (say, red) from $\alpha$ to $u$. Next, complete this construction: along $\mathfrak{p}_{\cup}$, propagate the blue path from $\alpha$ towards $u$ with probability $s$ on each edge, stopping at the first time when one red edge is not selected. Do the symmetrical construction to propagate the red path from $\alpha$ towards $i$. Finally, to each double-colored vertex, attach independent realizations of model $GW(\lambda',s)$, and to each single-colored vertex, attach independent realizations of model $GW(\lambda')$.
\end{itemize}
Note that these constructions lead to at most one path $\mathfrak{p}_{\cup}$ between $i$ and $u$ in $\cB_{G_1}(i,d) \cup \cB_{G_2}(u,d)$, so a fortiori in $\cB_{G_1}(i,d) \cap \cB_{G_2}(u,d)$. Denote by  $\mathfrak{p}_{\cap}$ this hypothetical path (cf. figure \ref{fig_parrallel_construction_bis}). We then distinguish between several cases.
\begin{figure}[H]
\centering
\includegraphics[scale=0.85]{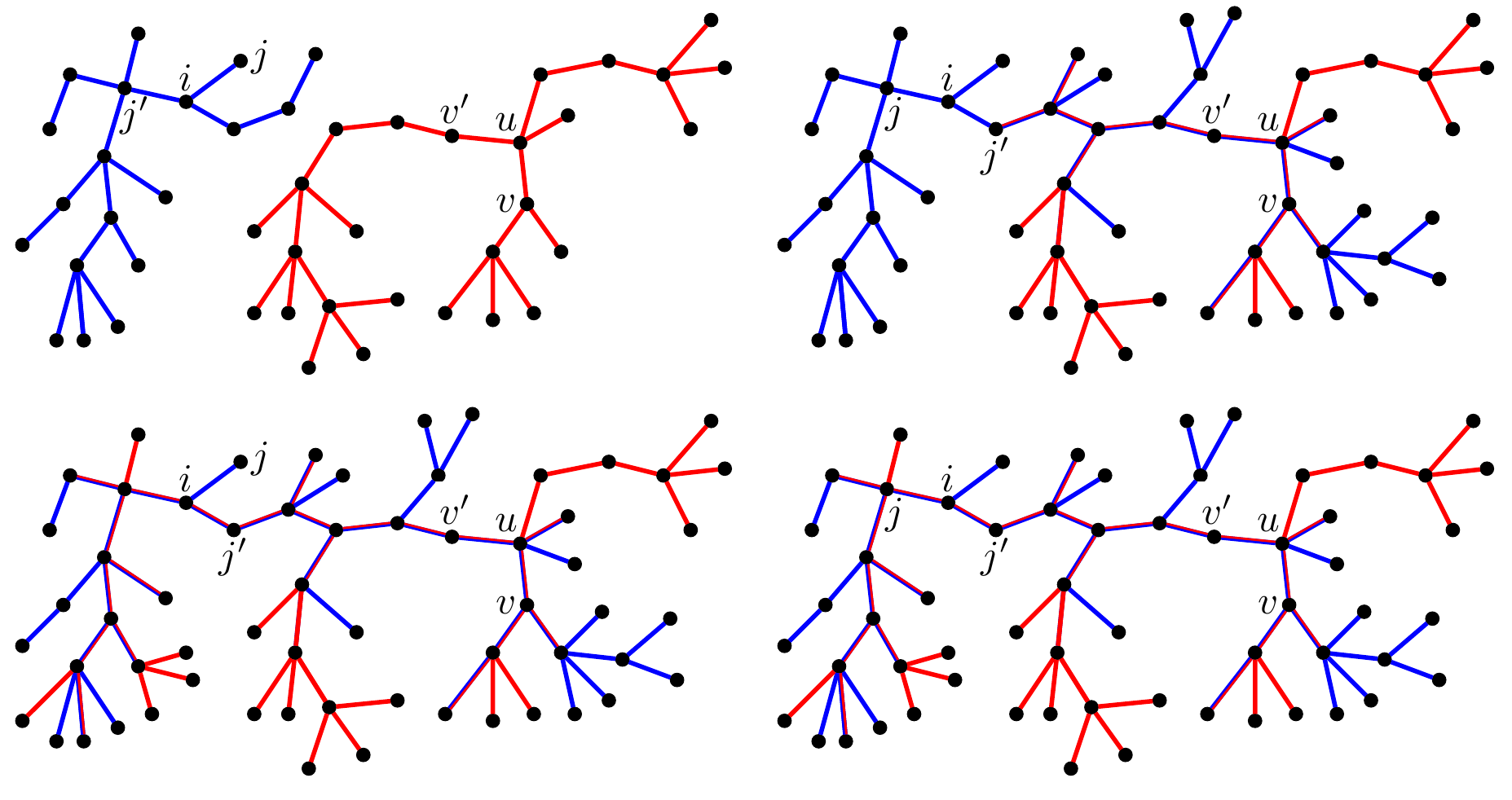}
\caption{\label{fig_parrallel_construction_bis} Possible realizations of $\cB_{G_1}(i,d)$ (blue) and $\cB_{G_2}(u,d)$ (red), with distinct cases $(i)$ (top left), $(ii)$ (top right), $(iii.a)$ (bottom left) and $(iii.b)$ (bottom right).}
\end{figure}\underline{Case $(i)$}: $\delta_{G_{\cup}}(i,u)>2d$ (figure \ref{fig_parrallel_construction_bis}, top left), i.e. $\cB_{G_1}(i,d)\cap\cB_{G_2}(u,d)=\varnothing$. The  construction gives a coupling with two independent trees from model $GW(\lambda)$. By assumption $\gamma(\lambda)<\lambda s$, the probability that there exist $ j$ in $\mathcal{N}_{G_1}(i)$ and $ v$ in $\mathcal{N}_{G_2}(u)$ such that $\mathcal{W}_{d-1}(j \leftarrow i, v \leftarrow u)>\gamma^{d-1}$ is upper bounded by $O\left(d_{\mathrm{max}}^2 \exp\left(-n^{\eps}\right)\right)$, following Remark \ref{rem:1.5}. Hence $i$ is matched to $u$ with at most this probability.\\

\underline{Case $(ii)$}: $\delta_{G_{\cup}}(i,u) \leq 2d$ but $\mathfrak{p}_{\cap}$ does not exist (see figure \ref{fig_parrallel_construction_bis}, top right). Take $v \neq v'$ two neighbors of $u$ and $j \neq j'$ two neighbors of $i$. Then (at least) one of these vertices is not on $\mathfrak{p}_{\cup}$ (e.g. vertex $j$ on figure \ref{fig_parrallel_construction_bis}): the downstream tree from this vertex is independent from every other neighborhood in the other graph. They can be coupled with model $GW(\lambda)$, and the same bound as in case $(i)$ holds.\\

Now assume that $\mathfrak{p}_{\cap}$ exists, and let $v \neq v'$ two neighbors of $u$ and $j \neq j'$ two neighbors of $i$. \underline{Case $(iii.a)$}: At least one of the edges $(i,j),(i,j'),(u,v),(u,v')$  is not in $G_{\cap}$ (e.g. edge $(i,j)$ on figure \ref{fig_parrallel_construction_bis}, bottom left): again, the same argument applies. \underline{Case $(iii.b)$}:  Edges $(i,j),(i,j'),(u,v),(u,v')$  are all in $G_{\cap}$ (see figure \ref{fig_parrallel_construction_bis}, bottom right). Then one pair of vertices (say $(j',v')$ as on figure \ref{fig_parrallel_construction_bis}) can be on $\mathfrak{p}_{\cap}$ and bring a high $\cW_{d-1}(j' \leftarrow i, v' \leftarrow u)>\gamma^{d-1}$ matching weight, if their descendants  spread over a great part of the intersection. In that case,  since $j$ and $v$ can't be on $\mathfrak{p}_{\cap}$, the associated downstream trees are independent, and again $\cW_{d-1}(j \leftarrow i, v \leftarrow u)<\gamma^{d-1}$ with high probability. \\
The remaining case to be considered is that of matches $(j,v')$ and $(j',v)$, with $j',v'$ on $\mathfrak{p}_{\cap}$. All trees involved are then correlated. However, the coupling construction induces a coupling of the two pairs of $(d-1)-$neighborhoods (from $(j,v')$ and from $(j',v)$, see figure \ref{fig_parrallel_construction_bis}) with two pairs of trees from model $GW(\lambda', s, \delta)$ where $\delta = \left|\mathfrak{p}_{\cap}\right|$. The Theorem assumes  $\gamma(\lambda,s,\delta)<\lambda s$ so that, by Theorem \ref{lambda_close_to_1_delta}, the probability that $\mathcal{W}_{d-1}(j \leftarrow i, v' \leftarrow u)>\gamma^{d-1}$ and $\mathcal{W}_{d-1}(j' \leftarrow i, v \leftarrow u)>\gamma^{d-1}$ is upper bounded by $O\left(\exp\left(-n^{\eps}\right)\right)$.\\

Thus, for $i$ fixed, one has 
$$
\mathbb{P}\left(\exists u  \neq i, \; (i,u) \in \cS\right) \leq 1-\dP\left(C_{\cup,i,2d}\right)+ n \times \dP\left(C_{\cup,i,2d}\right) \times  d_{\mathrm{max}}^2 \times O\left(\exp\left(-n^{\eps}\right) \right) = o(1).
$$
The Theorem then follows by appealing to Markov's inequality. \end{proof}
\end{document}